\documentclass{amsart}
\usepackage{mathrsfs}
\usepackage{graphicx}
\usepackage{amsmath}
\usepackage{amscd}
\usepackage{cite}
\usepackage{hyperref}
\usepackage{epsfig}
\usepackage{subfigure}
\usepackage{caption}
\usepackage{tikz}
\usepackage{float}
\usepackage{amssymb}
\usepackage{amsfonts}
\usepackage[all,cmtip]{xy}
\usepackage{appendix}
\usepackage{bm}

\newtheorem{Example}{Example}[section]

\newtheorem{Theorem}{Theorem}[section]
\newtheorem{Theorem/Definition}{Theorem/Definition}[section]
\newtheorem{Proposition}{Proposition}[section]
\newtheorem{Lemma}{Lemma}[section]
\newtheorem{Corollary}{Corollary}[section]
\newtheorem{Conjecture}{Conjecture}[section]

\newcommand{\pd}{\partial}

\newcommand{\bC}{{\mathbb C}}

\newcommand{\bZ}{{\mathbb Z}}

\newcommand{\cC}{{\mathcal C}}

\newcommand{\cF}{{\mathcal F}}

\newcommand{\cH}{{\mathcal H}}

\newcommand{\cM}{{\mathcal M}}

\newcommand{\half}{\frac{1}{2}}

\newcommand{\Mbar}{\overline{\cM}}

\newcommand{\wA}{{\widehat A}}

\newcommand{\tK}{{\widetilde K}}

\newcommand{\tB}{{\widetilde{B}}}

\newcommand{\tD}{{\widetilde{D}}}
\newcommand{\tE}{{\widetilde{E}}}

\newcommand{\hy}{{\hat y}}
\newcommand{\tomega}{{\tilde{\omega}}}

\newcommand{\be}{\begin{equation}}
	\newcommand{\ee}{\end{equation}}
\newcommand{\bea}{\begin{eqnarray}}
	\newcommand{\ben}{\begin{eqnarray*}}
		\newcommand{\een}{\end{eqnarray*}}
	\newcommand{\eea}{\end{eqnarray}}

\DeclareMathOperator{\Res}{Res}
\DeclareMathOperator{\Pf}{Pf}

\usepackage{color}

\definecolor{yellow}{rgb}{1,1,0}
\definecolor{orange}{rgb}{1,.7,0}
\definecolor{red}{rgb}{1,0,0}
\definecolor{green}{rgb}{0,1,1}
\definecolor{white}{rgb}{1,1,1}

\definecolor{A}{rgb}{.75,1,.75}

\theoremstyle{remark}
\newtheorem{Remark}{Remark}[section]

\begin{document}
	
	\newtheorem{myDef}{Definition}
	\newtheorem{thm}{Theorem}
	\newtheorem{eqn}{equation}

	\title[Emergent Geometry of Generalized BGW Tau-Functions]
	{BKP-Affine Coordinates and Emergent Geometry of Generalized Br\'ezin-Gross-Witten Tau-Functions}

	\author{Zhiyuan Wang}
	\address{School of Mathematics and Statistics,
		Huazhong University of Science and Technology,
		Wuhan, 430074, China}
	\email{wangzy23@hust.edu.cn}

	\author{Chenglang Yang}
	\address{Hua Loo-Keng Center for Mathematical Sciences,
		Academy of Mathematics and Systems Science,
		Chinese Academy of Sciences,
		Beijing, 100190, China}
	\email{yangcl@pku.edu.cn}

	\author{Qingsheng Zhang}
	\address{School of Sciences,
		Great Bay University,
          Dongguan, 523000, China}
	\email{zhangqingsheng@gbu.edu.cn}

	\begin{abstract}
		
		Following Zhou's framework,
		we consider the emergent geometry of the generalized Br\'ezin-Gross-Witten models
		whose partition functions are known to be a family of tau-functions of the BKP hierarchy.
		More precisely,
		we construct a spectral curve together with its special deformation,
		and show that the Eynard-Orantin topological recursion on this spectral curve
		emerges naturally from the Virasoro constraints for the generalized BGW tau-functions.
		Moreover,
		we give the explicit expressions for the BKP-affine coordinates of these tau-functions
		and their generating series.
		The BKP-affine coordinates and the topological recursion provide two different approaches
		towards the concrete computations of the connected $n$-point functions.
		Finally,
		we show that the quantum spectral curve of type $B$ in the sense of Gukov-Su{\l}kowski
		emerges from the BKP-affine coordinates and Eynard-Orantin topological recursion.
		
	\end{abstract}
	
	\maketitle

	\tableofcontents

	\section{Introduction}

	\subsection{Backgrounds and overview}
	
	The reduction/reconstruction and the emergence approach are two different philosophies in science,
	see e.g. \cite{an, la}.
	The later is now widely used in the study of condensed-matter physics.
	In recent several years,
	J. Zhou has initiated a series of mathematical works to introduce the emergent point of view
	into the study of Gromov-Witten type theories and mirror symmetry,
	see \cite{zhou1, zhou9, zhou10, zhou8, zhou7}.
	His framework is called the emergent geometry.
	A typical question in this framework is that,
	given a Gromov-Witten type theory,
	how can one find the associated B-model geometry
	which plays the role of a mirror?
	In his works,
	Zhou has examined several examples and
	proposed that a spectral curve (together with a deformation called the special deformation)
	emerges naturally from the connected one-point functions of the partition function,
	and the Eynard-Orantin topological recursion \cite{eo} emerges naturally from the Virasoro constraints.
	It is worth mentioning that in these examples
	the partition function can be uniquely determined by the Virasoro constraints,
	and thus Zhou took the Virasoro constraints as the starting point of the emergent geometry.

	Moreover,
	the partition functions in Zhou's examples are all tau-functions of the KP hierarchy.
	There are two main advantages to study theories specified by such tau-functions.
	The first advantage is that a tau-function $\tau(\bm t)$ of the KP hierarchy
	satisfying the initial value condition $\tau(0)=1$
	can be described uniquely by a set of (complex) numbers $\{a_{n,m}\}_{n,m\geq 0}$
	called the (KP-)affine coordinates \cite{zhou1, by},
	which are natural coordinates on the big cell of the Sato Grassmannian \cite{sa, sw, djm}.
	One can carry out a lot of explicit computations using the affine coordinates
	of such a tau-function.
	Another advantage is that,
	the Sato Grassmannian (i.e., the set of all tau-functions of the KP hierarchy)
	is an infinite-dimensional homogeneous space,
	and thus one can find the duality among different tau-functions
	(and the theories specified by them)
	using the action of the infinite-dimensional Lie group $\widehat{GL(\infty)}$,
	see e.g. \cite{zhou5, zhou7}.

	Now in this work,
	we give an example of the emergent geometry associated to a tau-function of the BKP hierarchy.
	The BKP hierarchy is an integrable system introduced by the Kyoto School \cite{djkm, jm}
	which shares a lot of common features with the KP hierarchy.
	In particular,
	tau-functions of the BKP hierarchy can be represented as summations
	of Schur $Q$-functions \cite{sch}, see \cite{yo};
	and a tau-function $\tau(\bm t)$ satisfying $\tau(0)=1$ can be described uniquely
	in terms of its BKP-affine coordinates on the big cell of the isotropic Sato Grassmannian,
	see \cite[\S 7]{hb} and \cite{wy}.
	See e.g. \cite{jwy, kv, or, va, al5} for more about the BKP hierarchy.

	Our object of study in this work is the generalized Br\'ezin-Gross-Witten (BGW) models.
	The partition functions of the generalized BGW models \cite{mms}
	are known to be a family of tau-functions of the BKP hierarchy
	indexed by a complex parameter $N\in \bC$  \cite{mms, al5}.
	When $N=0$,
	this model becomes the original BGW model \cite{bg, gw},
	which is known to be related to certain intersection numbers on
	the Deligne-Mumford moduli spaces of stable curves \cite{dm, kn},
	see e.g. \cite{no, cgg} and \cite{yz}.
	In this work,
	we show that the spectral curve together with the Eynard-Orantin topological recursion on it
	emerges naturally from the Virasoro constraints \cite{al2} for these tau-functions.
	Moreover,
	we will give the explicit expressions of the BKP-affine coordinates of these tau-functions,
	and show that the quantum spectral curve (in the sense of Gukov-Su{\l}kowski \cite{gs})
	associated to the spectral curve
	emerges naturally from the BKP-affine coordinates and the E-O topological recursion.
	We summarize our plan of this paper in the following diagram:
	
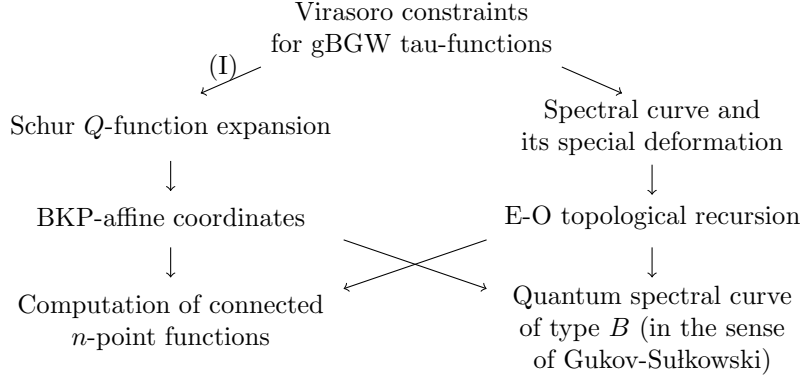
\begin{equation*}
		\begin{tikzpicture}
			\node [align=center,align=center] at (0,0.3) {Virasoro constraints\\for gBGW tau-functions};
			\draw [->] (-2,-0.2) -- (-2.8,-0.55);
			\node [align=center,align=center] at (-2.5,-0.2) {(I)};
			\node [align=center,align=center] at (-3.2,-1) {Schur $Q$-function expansion};
			\draw [->] (2,-0.2) -- (2.8,-0.55);
			\node [align=center,align=center] at (3.2,-1) {Spectral curve and\\its special deformation};
			\draw [->] (-3.2,-1.45) -- (-3.2,-1.85);
			\draw [->] (3.2,-1.5) -- (3.2,-1.9);
			\node [align=center,align=center] at (-3.2,-2.2) {BKP-affine coordinates};
			\node [align=center,align=center] at (3.2,-2.2) {E-O topological recursion};
			\draw [->] (-3.2,-2.6) -- (-3.2,-3);
			\draw [->] (3.2,-2.6) -- (3.2,-3);
			\node [align=center,align=center] at (-3.2,-3.6) {Computation of connected\\$n$-point functions};
			\node [align=center,align=center] at (3.2,-3.7) {Quantum spectral curve\\of type $B$ (in the sense\\of Gukov-Su{\l}kowski)};
			\draw [->] (-0.9,-2.5) -- (1,-3.15);
			\draw [->] (1,-2.5) -- (-0.9,-3.15);
		\end{tikzpicture}
	\end{equation*}
	Here the Virasoro constraints for the generalized BGW tau-functions
	were derived by Alexandrov in \cite{al2},
	and Step (I) has been accomplished by X. Liu and the second author in \cite{ly}
	(see also \cite{al3} for another approach to the Schur $Q$-function expansion).
	All other steps will be addressed in this work.

	\subsection{Description of main results}

	Now we briefly summarize our main results in this work.

	Denote by $\tau_{\text{BGW}}^{(N)}(\bm t)$ the generalized BGW tau-function
	indexed by the parameter $N\in \bC$,
	where $\bm t = (t_1,t_3,t_5,\cdots)$ are the coupling constants.
	Then the free energy $\log \tau_{\text{BGW}}^{(N)}$ admits a genus expansion
	(see e.g. \cite[(112)]{al2}):
	\begin{equation*}
		\log \tau_{\text{BGW}}^{(N)}
		(\hbar^{-1}\bm t;\hbar) = \sum_{g=0}^\infty \hbar^{2g-2} \cF_g^{\text{gBGW}}(\bm t;S),
	\end{equation*}
	where $S = \hbar\cdot N$ is a parameter for the generalized BGW models.
	Following Zhou's idea,
	we define the special deformation of the spectral curve
	to be the following formal series in $x$ where the coefficients are
	the times variables $\bm t$ (with a dilaton shift)
	and the one-point functions of genus zero:
	\begin{equation*}
		y = \sum_{n\geq 0}(2n+1)(t_{2n+1}-\delta_{n,0}) x^{2n}
		+\sum_{n\geq 0} \frac{\pd \cF_0^{\text{gBGW}}(\bm t)}{\pd t_{2n+1}} \cdot  x^{-2n-2}.
	\end{equation*}
	This is a family of curves on the $(x,y)$-plane parametrized by $\bm t$.
	By restricting to $\bm t =0$,
	we obtain a plane curve which plays the role of spectral curve
	(see \S \ref{sec-curve}):
	\be
	\label{eq-intro-curve}
	x^2y^2 = x^2 + S^2.
	\ee
	We will also discuss the quantum deformation theory of this spectral curve.
	More precisely,
	we show that the Virasoro constraints of genus zero can be encoded in the above special deformation,
	and the Virasoro constraints of all genera can be encoded in  a suitable quantization
	of the special deformation,
	see \S \ref{sec-special}-\S \ref{sec-QDT} for details.

	Then we consider the Eynard-Orantin topological recursion on this spectral curve.
	Notice that the above spectral curve degenerates when $S=0$,
	and in this case one cannot apply the E-O topological recursion directly.
	Nevertheless,
	we can treat $S$ as a formal variable and carry out the E-O recursion first,
	and then the $n$-point functions for the original BGW model can be obtained by restricting to $S=0$.
	See also \cite{dn} for a topological recursion of the original BGW model.
	
	In this work,
	we show that the Eynard-Orantin topological recursion emerges naturally from
	the Virasoro constraints for the generalized BGW tau-functions.
	Our main theorem is the following
	(see \S \ref{sec-bergman-eo} for details):
	\begin{Theorem}
		\label{thm-intro-1}
		Denote by $B(x_1,x_2)$ the following symmetric bi-differential
		on the spectral curve \eqref{eq-intro-curve}:
		\begin{equation*}
			B(x_1,x_2) = \Big( \frac{1}{(x_1-x_2)^2}
			+ W_{0,2}^{\text{gBGW}}(x_1,x_2) \Big)dx_1dx_2,
		\end{equation*}
		where $W_{g,n}^{\text{gBGW}}$ is the connected $n$-point function
		of genus $g$:
		\begin{equation*}
			W_{g,n}^{\text{gBGW}}(x_1,\cdots,x_n) =
			\sum_{i_1,\cdots,i_n\geq 0}
			\frac{\pd^n \cF_g(\bm t;S)}{ \pd t_{2i_1+1} \cdots \pd t_{2i_n+1}}
			\Big|_{\bm t =0} \cdot
			x_1^{-2i_1-2}\cdots x_n^{-2i_n-2}.
		\end{equation*}
		Let $\omega_{g,n}$ (where $g\geq 0,n\geq 1$) be the Eynard-Orantin invariants for
		the spectral curve \eqref{eq-intro-curve} and Bergman kernel $B$,
		then for $(g,n)$ with $2g-2+n>0$ we have:
		\begin{equation*}
			\omega_{g,n} = (-1)^n \cdot
			W_{g,n}^{\text{gBGW}}(x_1,\cdots, x_n) dx_1\cdots dx_n.
		\end{equation*}
	\end{Theorem}
	
	Moreover,
	the above topological recursion can also be reformulated
	in terms of another Bergman kernel
	(see \S \ref{sec-berg-B} for details):
	\begin{equation*}
		\tB(x_1,x_2) = \Big( \half\big(\frac{1}{(x_1-x_2)^2} +\frac{1}{(x_1+x_2)^2}\big)
		+ W_{0,2}^{\text{gBGW}}(x_1,x_2) \Big)dx_1dx_2.
	\end{equation*}
	In literatures a symmetric bi-differential of the above form is supposed to be
	the Bergman kernel for the topological recursion associated to
	a tau-function of the BKP hierarchy,
	see e.g. \cite{zhou2, as, gkl}.

	\begin{Remark}
		In \cite{al2},
		Alexandrov has conjectured a topological recursion on another spectral curve
		for the generalized BGW models
		(see \cite[(128); Conjecture 3.7]{al2}):
		\begin{equation*}
			x^2y^2 -x -\frac{S^2}{4} =0.
		\end{equation*}
		His method is to look for an operator annihilating the principal specialization of the tau-function
		and then take the semi-classical limit.
		This curve looks different from \eqref{eq-intro-curve},
		since he made a change of variable $x=\lambda^2$ in the principal specialization
		to represent it in terms of the modified Bessel function
		(see also \cite[\S 3.2.3]{mms}).
		As a byproduct of Theorem \ref{thm-intro-1},
			we will give a simple proof of Alexandrov's conjecture
			by comparing the E-O topological recursions on these two curves,
			see \S \ref{sec-bergman-eo}.
	\end{Remark}

	Besides the Virasoro constraints and topological recursion,
	there is an alternative way to compute the connected $n$-point functions.
	Using the Schur $Q$-expansion given in \cite{ly, al3, al4},
	we are able to write down the explicit formulas for the BKP-affine coordinates
	for the generalized BGW tau-functions.
	Then we can apply a formula proved in \cite{wy} to obtain the following
	(see \S \ref{sec-affinecoord-gbgw} for details):
	\begin{Theorem}
		The connected $n$-point functions associated  to the generalized BGW tau-functions are given by:
		\be
		\label{eq-intro-npt}
		\begin{split}
			&\sum_{i_1,\cdots,i_n> 0: \text{ odd}}
			\frac{\pd^n \log\tau_{\text{BGW}}^{(N)}(\bm t/2)}{\pd t_{i_1}\cdots \pd t_{i_n}} \bigg|_{\bm t=0}
			\cdot x_1^{-i_1}\cdots x_n^{-i_n}
			=
			-\delta_{n,2} \cdot
			\frac{x_1x_2(x_2^2+x_1^2)}{2(x_1^2-x_2^2)^2}\\
			&\qquad\qquad\qquad\qquad\qquad\qquad
			- 2^{n-1}  \cdot \Big[ \sum_{ \sigma: \text{ $n$-cycle} }
			\prod_{i=1}^n \xi( x_{\sigma(i)}, - x_{\sigma(i+1)})
			\Big]_{\text{odd}},
		\end{split}
		\ee
		for $n\geq 1$,
		where $[\cdot]_{\text{odd}}$ means taking the terms of odd degrees in every $x_i$,
		and we use the convention $\sigma(n+1):=\sigma(1)$.
		The function $\xi$ is:
		\begin{equation*}
			\begin{split}
				\xi( x_{\sigma(i)}, - x_{\sigma(i+1)}) =
				\begin{cases}
					A^{(N)} (x_{\sigma(i)}, - x_{\sigma(i+1)}),
					& \sigma(i) = \sigma(i+1);\\
					\wA^{(N)} ( x_{\sigma(i)}, - x_{\sigma(i+1)}),
					& \sigma(i)<\sigma(i+1);\\
					-\wA^{(N)}( - x_{\sigma(i+1)} , x_{\sigma(i)}),
					& \sigma(i)>\sigma(i+1),
				\end{cases}
			\end{split}
		\end{equation*}
		where $A^{(N)}$ and $\wA^{(N)}$ are the generating series of the BKP-affine coordinates
		(i.e., the fermionic $2$-point functions) whose explicit formulas are:
		\begin{equation*}
			\begin{split}
				&A^{(N)}(w,x)
				= \frac{w-x+\Phi_1^{(N)}(-x)\Phi_2^{(N)}(-w)-\Phi_1^{(N)}(-w)\Phi_2^{(N)}(-x)}{4(w+x)},\\
				&\wA^{(N)}(w,x)
				= \frac{\Phi_1^{(N)}(-x)\Phi_2^{(N)}(-w)-\Phi_1^{(N)}(-w)\Phi_2^{(N)}(-x)}{4(w+x)},
			\end{split}
		\end{equation*}
		where $\Phi_1^{(N)}$, $\Phi_2^{(N)}$ are the formal Laurent series:
		\begin{equation*}
			\begin{split}
				&\Phi_1^{(N)} (z) = 1+
				\sum_{k=1}^\infty \frac{(-\hbar)^k }{8^k \cdot k!}\cdot
				\prod_{i=1}^k\Big( 4N^2 - (2i-1)^2 \Big)
				\cdot z^{-k},\\
				&\Phi_2^{(N)} (z) = z+
				\sum_{k=1}^\infty \frac{(-\hbar)^k }{8^k \cdot k!}\cdot
				\prod_{i=1}^k\Big( 4(1-N)^2 - (2i-1)^2 \Big)
				\cdot z^{1-k}.
			\end{split}
		\end{equation*}
	\end{Theorem}

	It is interesting to compare the above explicit formula with other methods towards concrete computations
		of connected $n$-point functions of the generalized BGW models.
		Notice that \eqref{eq-intro-npt} is a formula for
		connected $n$-point functions at all genera,
		while the Eynard-Orantin topological recursion is efficient only at small genus.
		In literatures,
		there are also formulas for these connected $n$-point functions
		in terms of the KP-affine coordinates,
		see Dubrovin-Yang-Zagier \cite[Theorem 5]{dyz}.
		Their formula also represents these connected $n$-point functions as
		summations over $n$-cycles.
		The BKP-affine coordinates of the generalized BGW tau-functions are simpler than the KP-affine coordinates,
		while the formula \cite[(43)]{dyz} is simpler than \eqref{eq-intro-npt} in the sense that
		\cite[(43)]{dyz} does not require taking odd parts.
		Formulas of this type are known to be efficient at large genus,
		for example,
		Bertola-Dubrovin-Yang \cite{bdy} have presented some examples of the computations
		of intersection numbers on $\Mbar_{g,n}$ at arbitrary $g$ using formulas of this type.

	Finally,
	we consider the emergence of the quantum spectral curve of type $B$
	for the plane curve \eqref{eq-intro-curve}.
	In \cite{gs} Gukov and Su{\l}kowski proposed a conjectural construction of
	the quantum spectral curve for a classical plane curve
	using Eynard-Orantin topological recursion.
	It is well-known that when the generating series of the E-O invariants
	on the classical curve is a tau-function of the KP hierarchy,
	finding the quantum spectral curve (i.e., the Schr\"odinger equation \cite[(1.10)]{gs})
	is equivalent to finding an operator which annihilates the principal specialization
	of the tau-function.
	When considering a tau-function of the BKP hierarchy,
	the principal specialization is a slightly different with
	that in case of KP hierarchy,
	and the definition of the Baker-Akhiezer function \cite[(1.11);(2.4)-(2.7)]{gs}
	may need some simple modifications.
	Nevertheless,
	in what follows we will refer the following two facts as
	a quantum spectral curve of type $B$ in the sense of Gukov-Su{\l}kowski:
	\begin{itemize}
		\item[1)]
		The connected $n$-point functions of a tau-function $\tau(\bm t)$ of the BKP hierarchy
		can be reconstructed from the Eynard-Orantin topological recursion on a plane curve $\cC$;
		\item[2)]
		There is an operator $P$ annihilating the following principal specialization:
		\begin{equation*}
			P\Big(
			\tau\big(
			-\frac{2}{z}, -\frac{2}{3z^3}, -\frac{2}{5z^5},\cdots \big) \Big) =0,
		\end{equation*}
		such that the semi-classical limit of $P$ gives the defining equation of $\cC$.
	\end{itemize}

	In \cite{jwy},
	a method of deriving Kac-Schwarz operators \cite{ks, sc} and quantum spectral curve of type $B$
	using BKP-affine coordinates has been proposed.
	Now we apply this method to the case of generalized BGW models
	using the explicit expressions of BKP-affine coordinates computed in \S \ref{sec-affinecoord-gbgw}.
	We show that the operator
	\begin{equation*}
		P=\hbar^3 \Big( (z\pd_z+\half )^2 -N^2\Big)
		\Big( \pd_z -\frac{\hbar}{2z^2} \big(
		(z\pd_z-\frac{1}{2})^2 -N^2
		\big) \Big)
	\end{equation*}
	gives the quantum spectral curve of type $B$ in the sense of Gukov-Su{\l}kowski
	for the classical curve \eqref{eq-intro-curve},
	see \S \ref{sec-qsc-B} for details.
	It is worth pointing out that
	in \cite{al3} Alexandrov has found that such an operator annihilates
	the principal specialization (without knowing the topological recursion) using a different method.

	The rest of this paper is arranged as follows.
	In \S \ref{sec-pre} we recall some preliminaries of the generalized BGW tau-functions,
	including the Virasoro constraints.
	In \S \ref{sec-affinecoord-gbgw} we compute the BKP-affine coordinates
	and derive the formula \eqref{eq-intro-npt}.
	In \S \ref{sec-curve} we construct the spectral curve \eqref{eq-intro-curve}
	and discuss the quantum deformation theory.
	In \S \ref{sec-eo} we show that the Eynard-Orantin topological recursion on the curve \eqref{eq-intro-curve}
	emerges naturally from the Virasoro constraints of the generalized BGW tau-functions.
	And finally in \S \ref{sec-qsc-B},
	we show the emergence of the quantum spectral curve of type $B$
	from the BKP-affine coordinates and the above topological recursion.

	\section{Preliminaries of Generalized Br\'ezin-Gross-Witten Tau-Functions}
	\label{sec-pre}
	
	In this section,
	we first briefly review
	some basics of the generalized Br\'ezin-Gross-Witten tau-functions $\tau_{\text{BKP}}^{(N)}(\bm t)$,
	including the Virasoro constraints,
	recursion for $n$-point functions,
	and Schur $Q$-function expansion.

	\subsection{Generalized BGW tau-functions and Virasoro constraints}
	
	The generalized BGW model was introduced by Mironov-Morozov-Semenoff in \cite{mms}
	as a family of matrix integrations indexed by a parameter $N$
	(where $N$ is a complex number and is not to be confused with the size of the matrices).
	For each $N \in \bC$ the partition function $\tau_{\text{BGW}}^{(N)} (\bm t;\hbar)$ is a tau-function
	of the KdV hierarchy with time variables $\bm t = (t_1,t_3,t_5,\cdots)$.
	And for $N=0$,
	the partition function $\tau_{\text{BGW}}^{(0)}$
	is the original BGW tau-function \cite{bg, gw}.

	In \cite{al2},
	Alexandrov showed that for every $N$
	the partition function $\tau_{\text{BGW}}^{(N)} (\bm t;\hbar)$ is uniquely determined by
	the normalization condition
	\begin{equation*}
		\tau_{\text{BGW}}^{(N)} (0;\hbar) =1,
	\end{equation*}
	together with the Virasoro constraints
	\be
	\label{eq-Virasoro-1}
	L_n^{(N)} \tau_{\text{BGW}}^{(N)} (\bm t ; \hbar)=0,
	\qquad \forall n\geq 0,
	\ee
	where the Virasoro operators $\{L_n^{(N)}\}_{n\geq 0}$ are:
	\begin{equation*}
		L_n^{(N)} = -\frac{1}{2\hbar} \frac{\pd}{\pd t_{2n+1}} +
		\frac{1}{4} \sum_{\substack{a+b = 2n\\ a,b: \text{ odd}}} \frac{\pd^2}{\pd t_a \pd t_b}
		+ \half \sum_{\substack{k\geq 1\\k:\text{ odd}}} kt_k \frac{\pd}{\pd t_{k+2n}}
		+\frac{1-4N^2}{16}\delta_{n,0}.
	\end{equation*}
	which satisfy the Virasoro commutation relation:
	\begin{equation*}
		[L_m^{(N)},L_n^{(N)}] = (m-n)L_{m+n}^{(N)}, \qquad
		\forall m,n\geq 0.
	\end{equation*}

	The free energy associated to $\tau_{\text{BGW}}^{(N)}$ is the logarithm
	\begin{equation*}
		\cF^{(N)}(\bm t;\hbar) =
		\log \tau_{\text{BGW}}^{(N)}(\bm t;\hbar).
	\end{equation*}
	Denote:
	\be
	S = \hbar\cdot  N,
	\ee
	then the free energy $\cF^{(N)}$
	has a genus expansion (see e.g. \cite[(112)]{al2}):
	\begin{equation*}
		\cF^{(N)}(\hbar^{-1}\bm t;\hbar) = \sum_{g=0}^\infty \hbar^{2g-2} \cF_g^{\text{gBGW}}(\bm t;S),
	\end{equation*}
	where $\hbar^{-1}\bm t = (\hbar^{-1}t_1,\hbar^{-1}t_3,\hbar^{-1}t_5,\cdots)$.
	In what follows,
	we will denote:
	\be
	\label{eq-F-genusexp}
	\cF^{\text{gBGW}} (\bm t; S)
	= \sum_{g=0}^\infty \hbar^{2g-2} \cF_g^{\text{gBGW}}(\bm t;S) =
	\cF^{(N)} (\hbar^{-1}\bm t;\hbar),
	\ee
	and thus one has $\tau_{\text{BGW}}^{(N)} (\hbar^{-1}\bm t;\hbar)
	= e^{ \cF^{\text{gBGW}} (\bm t; S)}$.
	The Virasoro constraints \eqref{eq-Virasoro-1} now can be rewritten as
	\be
	\label{eq-Virasoro-2}
	L_n^{\text{gBGW}} \big( \exp\cF^{\text{gBGW}} (\bm t; S)\big) =0,
	\qquad \forall n\geq 0,
	\ee
	where:
	\begin{equation*}
		\begin{split}
			& L_0^{\text{gBGW}} = -\frac{1}{2} \frac{\pd}{\pd t_{1}} +
			\frac{1}{2} \sum_{\substack{k\geq 1\\k:\text{ odd}}} k t_k \frac{\pd}{\pd t_{k}}
			+\frac{1}{16} - \hbar^{-2}\frac{S^2}{4},\\
			&  L_n^{\text{gBGW}} = -\frac{1}{2} \frac{\pd}{\pd t_{2n+1}} +
			\frac{\hbar^2}{4} \sum_{\substack{a+b = 2n\\ a,b: \text{ odd}}} \frac{\pd^2}{\pd t_a \pd t_b}
			+ \frac{1}{2} \sum_{\substack{k\geq 1\\k:\text{ odd}}} k t_k \frac{\pd}{\pd t_{k+2n}},
			\quad n\geq 1.
		\end{split}
	\end{equation*}

	\begin{Remark}
		Our notations for the partition function and free energy
		differ from those in \cite{al2} by a rescaling $t_i \mapsto 2t_i$ for every $i$.
	\end{Remark}

	\subsection{Recursion for connected $n$-point functions}
	\label{sec-rec-npt}
	
	In this subsection,
	we reformulate the Virasoro constraints as
	recursions for the connected $n$-point functions.
	This has been done by Alexandrov in \cite[\S 3]{al2} (in different notations).

	Let $\langle  p_{\mu_1}\cdots p_{\mu_n} \rangle_g^c$ be the connected correlators defined by:
	\be
	\langle  p_{\mu_1} p_{\mu_2}\cdots p_{\mu_n} \rangle_g^c
	= \frac{\pd^n \cF_g^{\text{gBGW}}(\bm t;S)}{\pd t_{\mu_1} \pd t_{\mu_2}\cdots \pd t_{\mu_n}}
	\Big|_{\bm t=0},
	\qquad
	\mu_1,\cdots,\mu_n>0:\text{ odd},
	\ee
	then the free energy is of the form:
	\begin{equation*}
		\cF_g^{\text{gBGW}} (\bm t;S) = \sum_{n\geq 1}\sum_{\mu_i>0:\text{ odd}}
		\frac{1}{n!} \langle  p_{\mu_1}\cdots p_{\mu_n} \rangle_g^c
		\cdot t_{\mu_1} \cdots t_{\mu_n}.
	\end{equation*}
	One can rewrite the Virasoro constraints \eqref{eq-Virasoro-2} as
	the recursion for the correlators:
	\be
	\label{eq-rec-corr}
	\begin{split}
		&\langle p_{2k+1} p_{\mu_1} p_{\mu_2}\cdots p_{\mu_n} \rangle_g^c \\
		=& \half \sum_{\substack{a+b=2k\\a,b>0:\text{ odd}}}
		\Big(  \langle p_ap_b p_{\mu_1}\cdots p_{\mu_n} \rangle_{g-1}^c
		+ \sum_{\substack{g_1+g_2 = g\\ I\sqcup J = [n]}}
		\langle p_a p_{\mu_I} \rangle_{g_1}^c \langle p_b p_{\mu_J} \rangle_{g_2}^c \Big)\\
		&+ \sum_{i=1}^n \mu_i \langle p_{\mu_i+2k} p_{\mu_1}\cdots \hat p_{\mu_i} \cdots p_{\mu_n} \rangle_g^c,
		\qquad\qquad k\geq 0,
	\end{split}
	\ee
	together with the initial values:
	\be
	\begin{split}
		&\langle p_1 \rangle_0^c = -\half S^2,\qquad
		\langle p_1 \rangle_1^c = \frac{1}{8};\\
		&\langle p_1 \rangle_g^c =0,\qquad \forall g\geq 2,
	\end{split}
	\ee
	where $\hat p_{\mu_i}$ means deleting the term $p_{\mu_i}$.
	Here $[n] = \{1,2,\cdots,n\}$,
	and for a set of indices $I=\{i_1,i_2,\cdots,i_m\}\subset [n]$
	we denote $p_{\mu_I} = (p_{\mu_{i_1}},p_{\mu_{i_2}},\cdots,p_{\mu_{i_m}})$.
	In the case $g=0$, we use the convention
	$\langle p_a p_b p_{\mu_1}\cdots p_{\mu_k} \rangle_{-1}^c =0$
	in the right-hand side.

	Now consider the $n$-point function $W_{g,n}^{\text{gBGW}}$ of genus $g$ defined by:
	\be
	\begin{split}
		W_{g,n}^{\text{gBGW}}(x_1,\cdots,x_n)
		=&  \sum_{\mu_1,\cdots,\mu_n:\text{ odd}} \langle  p_{\mu_1}\cdots p_{\mu_n} \rangle_g^c
		\cdot
		x_1^{-\mu_1-1}\cdots x_n^{-\mu_n-1}\\
		=&
		\sum_{i_1,\cdots,i_n\geq 0}
		\frac{\pd^n \cF_g^{\text{gBGW}}(\bm t;S)}{ \pd t_{2i_1+1} \cdots \pd t_{2i_n+1}}
		\Big|_{\bm t =0} \cdot
		x_1^{-2i_1-2}\cdots x_n^{-2i_n-2},
	\end{split}
	\ee
	for $g\geq 0$ and $n\geq 1$,
	where $x_1,\cdots,x_n$ are some formal variables.
	For example,
	using the above recursion for correlators one easily finds that:
	\be
	\langle p_{2n+1} \rangle_0^c
	= \frac{(-1)^{n+1}}{2^{2n+1}} \cdot \frac{1}{n+1}\binom{2n}{n}\cdot S^{2n+2},
	\qquad \forall n\geq 0,
	\ee
	and thus:
	\be
	W_{0,1}^{\text{gBGW}}(x)
	= \sum_{n\geq 0}
	\frac{(-1)^{n+1}}{2^{2n+1}} \cdot \frac{1}{n+1}\binom{2n}{n}
	S^{2n+2} x^{-2n-2}
	= 1-\sqrt{1+\frac{S^2}{x^2}}.
	\ee
	
	\begin{Remark}
		Here our notation $W_{g,n}^{\text{gBGW}}$ differs from $W_{g,n}$ in \cite[\S 3]{al2} by:
		\begin{equation*}
			W_{g,n}^{\text{gBGW}}(x_1,\cdots,x_n)
			= 2^{-n}\cdot W_{g,n}(\frac{x_1^2}{4},\cdots,\frac{x_n^2}{4}).
		\end{equation*}
		The purpose of modifying the notations in this way is to make the discussions
		fit into the picture of emergent geometry
		(see \S \ref{sec-curve} - \S \ref{sec-eo} for details).
	\end{Remark}

	In the rest of this subsection,
	we rewrite the Virasoro constraints as
	the recursion for the $n$-point functions $W_{g,n}^{\text{gBGW}}(x_1,\cdots,x_n)$.
	By \eqref{eq-rec-corr} we have:
	\begin{equation*}
		\begin{split}
			&W_{g,n+1}^{\text{gBGW}} (x, x_1,\cdots,x_n)\\
			=& \sum_{k\geq 0} \sum_{\mu_1,\cdots,\mu_n:\text{ odd}}
			\langle p_{2k+1} p_{\mu_1} \cdots p_{\mu_n} \rangle_g \cdot
			x^{-2k-2} x_1^{-\mu_1-1}\cdots x_n^{-\mu_n-1}\\
			=& \sum_\mu \sum_{k=0}^\infty
			\sum_{a+b=2k} \half  \langle p_ap_bp_\mu \rangle_{g-1}\cdot
			x^{-a-1} x^{-b-1}
			x_1^{-\mu_1-1}\cdots x_n^{-\mu_n-1}\\
			&+
			\sum_\mu \sum_{k=0}^\infty
			\sum_{a+b=2k}\sum_{\substack{g_1+g_2 = g\\ I\sqcup J = [n]}}
			\half  \langle p_a p_{\mu_I} \rangle_{g_1}^c \langle p_b p_{\mu_J} \rangle_{g_2}^c
			\cdot x^{-a-1} x^{-b-1}
			x_1^{-\mu_1-1}\cdots x_n^{-\mu_n-1}\\
			&+
			\sum_\mu \sum_{k=0}^\infty \sum_{i=1}^n
			\mu_i \langle p_{\mu_i+2k} p_{\mu_1}\cdots \hat p_{\mu_i} \cdots p_{\mu_n} \rangle_g^c
			\cdot x^{-2k-2} x_1^{-\mu_1-1}\cdots x_n^{-\mu_n-1}.
		\end{split}
	\end{equation*}
	Now we conclude that for every $(g,n)\not= (0,0)$,
	\begin{equation*}
		\begin{split}
			&W_{g,n+1}^{\text{gBGW}} (x, x_1,\cdots,x_n) =
			\sum_{i=1}^n D_{x,x_i}W_{g,n}^{\text{gBGW}}(x_1,\cdots,x_n) + \\
			&  \half E_{x,u,v} \Big(
			W_{g-1,n+2}^{\text{gBGW}} (u,v,x_1,\cdots,x_n)
			+\sum_{\substack{g_1+g_2 = g\\ I\sqcup J = [n]}}
			W_{g_1,|I|+1}^{\text{gBGW}} (u,x_{I}) W_{g_2,|J|+1}^{\text{gBGW}} (v,x_{J}) \Big),
		\end{split}
	\end{equation*}
	where we denote $x_I = (x_{i_1},\cdots,x_{i_m})$ for $I=\{i_1,\cdots,i_m\}$,
	and
	\be
	E_{x,u,v} f(u,v) =\big( \lim_{v \to u} f(u,v) \big)\big|_{u=x},
	\ee
	and $D_{a,b}$ is an operator such that for every odd $m$,
	\begin{equation*}
		\begin{split}
			D_{x,x_i} (x_i^{-m-1}) =
			& \sum_{\substack{2k+\mu_i = m\\ \mu_i\text{ odd}; \text{ }k\geq 0}} \mu_i x^{-2k-2} x_i^{-\mu_i-1}\\
			=&
			\frac{1}{(x^2-x_i^2)^2}
			\Big( \frac{1}{x^{m-1}} - \frac{1}{x_i^{m-1}} +\frac{x_i^2}{x^{m+1}} -\frac{x^2}{x_i^{m+1}} \Big)
			+ \frac{m+1}{x_i^{m+1} (x^2-x_i^2)}
		\end{split}
	\end{equation*}
	So one can take the operator $D_{a,b}$ to be:
	\be
	D_{a,b} f(b) =
	\frac{a^2 + b^2}{(a^2 - b^2)^2}
	\Big( f(a) - f(b) \Big)
	-\frac{b}{a^2-b^2} \pd_b f(b).
	\ee
	Notice that the right-hand side of the above recursion involves
	terms $W_{0,1}^{\text{gBGW}}(x)$.
	We may move such terms to the left-hand side and rewrite
	the recursion as follows:
	\be
	\label{eq-rec-npt}
	\begin{split}
		&W_{g,n+1}^{\text{gBGW}} (x, x_1,\cdots,x_n)
		= \sum_{i=1}^n \tD_{x,x_i}W_{g,n}^{\text{gBGW}}(x_1,\cdots,x_n) +\\
		&\quad \half \tE_{x,u,v} \Big(
		W_{g-1,n+2}^{\text{gBGW}} (u,v,x_{[n]})
		+\sum_{\substack{g_1+g_2 = g\\ I\sqcup J = [n]}}^s
		W_{g_1,|I|+1}^{\text{gBGW}} (u,x_{I})  W_{g_2,|J|+1}^{\text{gBGW}} (v,x_{J}) \Big),
	\end{split}
	\ee
	where $\sum\limits^s$ means that we exclude the terms involving $W_{0,1}^{\text{gBGW}}$
	in this summation,
	and we use the notations:
	\be
	\tE_{x,u,v} = \frac{1}{1-W_{0,1}(x)} E_{x,u,v},
	\qquad
	\tD_{x,x_i} = \frac{1}{1-W_{0,1}(x)} D_{x,x_i}.
	\ee
	
	\begin{Example}
		\label{eg-w02}
		We can compute $W_{0,2}^{\text{gBGW}}$ using the above formula:
		\begin{equation*}
			\begin{split}
				W_{0,2}^{\text{gBGW}} (x,x_1)
				=& \frac{1}{1-W_{0,1}(x)} D_{x,x_1} W_{0,1}(x_1)\\
				=& \frac{1}{(x^2-x_1^2)^2}
				\Big( \frac{x^2+x_1^2+2S^2}{\sqrt{(1+\frac{S^2}{x^2})(1+\frac{S^2}{x_1^2})}}
				-x^2 -x_1^2 \Big)\\
				=& -\half S^2 x^{-2}x_1^{-2}
				+\frac{3}{8} S^4 x^{-2} x_1^{-4}
				+\frac{3}{8} S^4 x^{-4} x_1^{-2}
				-\frac{5}{16} S^6 x^{-2} x_1^{-6}\\
				& -\frac{3}{8} S^6 x^{-4} x_1^{-4}
				-\frac{5}{16} S^6 x^{-6} x_1^{-2}
				+\cdots.
			\end{split}
		\end{equation*}
	\end{Example}

	\subsection{Schur $Q$-function expansion and BKP-affine coordinates}
	
	In this subsection,
	we recall the Schur $Q$-function expansions of the generalized BGW tau-functions
	\cite{al3, al4, ly}.

	It is known that if $\tau(\bm t)$ is a tau-function of the KdV hierarchy
	with time variables $\bm t = (t_1,t_3,t_5,\cdots)$,
	then $\tau(\bm t/2)$ is a tau-function of the BKP hierarchy,
	see \cite{al5}.
	Therefore for every $N$,
	the partition function $\tau_{\text{BGW}}^{(N)}(\bm t/2)$
	is a tau-function of the BKP hierarchy.
	Moreover,
	every tau-function of the BKP hierarchy can be represented as
	a summation of Schur $Q$-functions, see \cite{yo};
	and see also \cite{mac, hh} for introductions to Schur $Q$-functions
	and the relation to projective representations of the symmetric groups $S_n$.
	Thus $\tau_{\text{BGW}}^{(N)}(\bm t/2)$ admits a Schur $Q$-function expansion
	for every $N\in \bC$.

	The following formula was conjectured by Alexandrov \cite{al4},
	and then proved in \cite{ly} and \cite{al3} by two different methods:
	\be
	\label{eq-Qexp-gbgw}
	\tau_{\text{BGW}}^{(N)} (\bm t)
	= \sum_{\lambda \in DP} \Big(\frac{\hbar}{16} \Big)^{|\lambda|}
	\cdot 2^{-l(\lambda)} \theta_\lambda
	Q_\lambda (\delta_{k,1}) Q_\lambda(\bm t),
	\ee
	where $DP$ is the set of all strict partitions
	$\lambda = (\lambda_1>\lambda_2>\cdots >\lambda_{l(\lambda)}>0)$,
	and $Q_\mu$ is the Schur $Q$-function indexed by $\mu\in DP$.
	And $\theta_\lambda$ is given by:
	\be
	\theta_\lambda = \prod_{j=1}^l \prod_{k=1}^{\lambda_j} \theta(k)
	\ee
	for $\lambda = (\lambda_1, \lambda_2,\cdots,\lambda_l) \in DP$,
	where $\theta$ is the following function on $\bZ_+$:
	\be
	\label{eq-def-theta}
	\theta (z) = (2z-1)^2 -4N^2.
	\ee
	This Schur $Q$-function expansion shows that $\tau_{\text{BGW}}^{(N)} (\bm t/2)$ is
	a hypergeometric tau-function \cite{or} of the BKP hierarchy.

	\begin{Remark}
		It is worth noting that in \cite{ly},
		Liu and the second author proved this Schur $Q$-function expansion
		by applying the Virasoro constraints \eqref{eq-Virasoro-1}.
	\end{Remark}

	\section{BKP-Affine Coordinates of Generalized BGW Tau-Functions}
	\label{sec-affinecoord-gbgw}
	
	In this section,
	we write down an explicit formula for the BKP-affine coordinates
	of the generalized BGW tau-functions using the Schur $Q$-function expansion.
	Then we apply the results in \cite{wy} to give a formula for the
	connected $n$-point functions.
	The BKP-affine coordinates will be useful in \S \ref{sec-qsc-B}.

	\subsection{BKP-affine coordinates of $\tau_{\text{BGW}}^{(N)}({\bf t}/2)$}
	
	In this subsection,
	we write down the BKP-affine coordinates for the generalized BGW tau-functions.

	First we recall some basics about BKP tau-functions and BKP-affine coordinates \cite{wy, hb}.
	BKP-affine coordinates are natural coordinates on the big cell of the isotropic Sato Grassmannian,
	see \cite[\S 7.3]{hb} for an introduction.
	Let $\tau(\bm t)$ be a tau-function of the BKP hierarchy with BKP-time variables
	$\bm t = (t_1,t_3,t_5,\cdots)$
	satisfying the initial value condition $\tau(0)=1$,
	then the coefficients of the Schur $Q$-functions in such an expansion
	are Pfaffians of the BKP-affine coordinates:
	\be
	\tau =
	\sum_{\mu \in DP} (-1)^{\lceil l(\mu)/2 \rceil} \cdot
	\Pf (a_{\mu_i,\mu_j})_{1\leq i,j\leq 2 \lceil l(\mu)/2 \rceil } \cdot Q_\mu(\bm t/2),
	\ee
	where $\{a_{n,m}\}_{n,m\geq 0}$ are the BKP-affine coordinates of $\tau(\bm t)$,
	satisfying the anti-symmetry condition:
	\be
	a_{n,m} = - a_{m,n}, \qquad \forall n,m\geq 0.
	\ee
	For $\mu = (\mu_1> \cdots>\mu_{l(\mu)}>0) \in DP$,
	we denote by $ l(\mu) $ the length of $\mu$,
	and here $\lceil \cdot \rceil$ is the ceiling function.
	For $\mu\in DP$ with $l(\mu)$ odd,
	we use the convention $\mu_{l(\mu)+1} = 0$.
	For example, the first a few terms of $\tau$ are:
	\begin{equation*}
		\begin{split}
			\tau =& 1+ \sum_{n>0}a_{0,n}\cdot Q_{(n)}(\bm t/2)
			+\sum_{m>n>0} a_{n,m}\cdot Q_{(m,n)}(\bm t/2)\\
			& + \sum_{m>n>l>0} (a_{n,m}a_{0,l} -a_{l,m}a_{0,n}+ a_{0,m}a_{l,n})Q_{(m,n,l)} (\bm t/2)\\
			& +\sum_{m>n>l>k>0} (a_{n,m}a_{k,l} -a_{l,m}a_{k,n}+ a_{k,m}a_{l,n})Q_{(m,n,l,k)} (\bm t/2) +\cdots.
		\end{split}
	\end{equation*}
	In particular,
	the BKP-affine coordinates $\{a_{n,m}\}_{n,m\geq 0}$ are exactly the coefficients
	of $Q_\mu$ with $l(\mu)\leq 2$.

	Now consider the case of generalized BGW tau-functions $\tau_{\text{BGW}}^{(N)} (\bm t)$.
	Using the Schur $Q$-function expansion \eqref{eq-Qexp-gbgw},
	we obtain the following:
	\begin{Proposition}
		The BKP-affine coordinates $\{a_{n,m}^{(N)}\}_{n,m\geq 0}$
		for the BKP tau-functions $\tau_{\text{BGW}}^{(N)} (\bm t/2)$ are given by
		$a_{0,0}^{(N)} = 0$, and:
		\be
		\label{eq-BKPaffinecor-gbgw}
		\begin{split}
			&a_{0,n}^{(N)} = -a_{n,0}^{(N)}
			= \frac{\hbar^n}{2^{3n+1}\cdot n!} \cdot \prod_{k=1}^n \theta(k),
			\qquad n>0;\\
			&a_{n,m}^{(N)} =
			\frac{\hbar^{m+n}}{2^{3m+3n+2} \cdot m!\cdot n!}\cdot \frac{m-n}{m+n}\cdot
			\prod_{j=1}^m \theta(j) \cdot \prod_{k=1}^n \theta(k),
			\qquad n,m>0,
		\end{split}
		\ee
		where $\theta$ is the function \eqref{eq-def-theta}.
	\end{Proposition}
	\begin{proof}
		Here we need the following combinatorial identity
		(see e.g. \cite[(56)]{mm}):
		\begin{equation*}
			Q_\lambda (\delta_{k,1})
			= \frac{2^{|\lambda|}}{\prod_{i=1}^{l(\lambda)}\lambda_i!} \cdot
			\prod_{i<j} \frac{\lambda_i - \lambda_j}{\lambda_i+\lambda_j}.
		\end{equation*}
		The conclusion is proved by plugging this identity into
		the coefficients of $Q_\lambda$ with length $l(\lambda)\leq 2$
		in the Schur $Q$-function expansion \eqref{eq-Qexp-gbgw}.
	\end{proof}
	
	In particular,
	by taking $N=0$ in \eqref{eq-BKPaffinecor-gbgw} we obtain the BKP-affine coordinates of
	the original BGW tau-function (see \cite[\S 6.3]{wy}):
	\begin{equation*}
		\begin{split}
			&a_{0,n}^{\text{BGW}} = -a_{n,0}^{\text{BGW}}
			= \frac{\hbar^n \cdot \big( (2n-1)!! \big)^2}{2^{3n+1}\cdot n!} ,
			\qquad n>0;\\
			&a_{n,m}^{\text{BGW}} =
			\frac{\hbar^{m+n} \cdot \big((2m-1)!!\cdot (2n-1)!!\big)^2}
			{2^{3m+3n+2} \cdot m!\cdot n!}\cdot \frac{m-n}{m+n},
			\qquad n,m>0.
		\end{split}
	\end{equation*}

		\begin{Remark}
			The KP-affine coordinates of the generalized BGW tau-functions
			have been computed by Dubrovin-Yang-Zagier in \cite[Proposition 3]{dyz}.
			See also \cite[\S 6.5]{zhou5}.
			The explicit expressions for the KP-affine coordinates of $\tau_{\text{BGW}}^{(N)} (\bm t)$ are more complicated
			than the expressions for the BKP-affine coordinates of $\tau_{\text{BGW}}^{(N)} (\bm t/2)$,
			but the generating series of these KP- and BKP-affine coordinates can both be represented
			in terms of the first two basis vectors of the corresponding point in the Sato Grassmannian
			(or equivalently, in terms of the modified Bessel functions) in a simple way,
			see \cite[(50)]{dyz} and \S \ref{sec-gen-BKPaff} respectively.
			
		\end{Remark}

	\subsection{A formula for connected $n$-point functions}
	
	Once knowing the BKP-affine coordinates,
	we are able to write down a formula for the connected bosonic $n$-point functions
	(in all genera) by applying \cite[Theorem 1.1]{wy}:
	\begin{Theorem}
		[\cite{wy}]
		Let $A(w,x),\wA(w,x)$ be the following generating series of the BKP-affine coordinates $\{a_{n,m}\}$
		of a BKP tau-function $\tau(\bm t)$ satisfying $\tau(0)=1$:
		\be
		\label{eq-gen-BKPaffine}
		\begin{split}
			&A(w,x) = \sum_{n,m>0} (-1)^{m+n+1} \cdot a_{n,m} w^{-n} x^{-m}
			- \sum_{n>0}\frac{(-1)^n}{2} \cdot a_{n,0} (w^{-n}-x^{-n}),\\
			&\wA(w,x) =  A(w,x) -\frac{1}{4} -\half\sum_{i=1}^\infty (-1)^{i} w^{-i} x^i.
		\end{split}
		\ee
		Then:
		\begin{equation*}
			\sum_{i> 0: \text{ odd}}
			\frac{\pd \log\tau(\bm t)}{\pd t_{i}} \bigg|_{\bm t=0}
			\cdot x^{-i}
			=A(-x,x),
		\end{equation*}
		and for $n\geq 2$,
		\be
		\label{eq-npt-general}
		\begin{split}
			&\sum_{i_1,\cdots,i_n> 0: \text{ odd}}
			\frac{\pd^n \log\tau(\bm t)}{\pd t_{i_1}\cdots \pd t_{i_n}} \bigg|_{\bm t=0}
			\cdot x_1^{-i_1}\cdots x_n^{-i_n}
			=
			-\delta_{n,2} \cdot i_{x_1,x_2}
			\frac{x_1x_2(x_2^2+x_1^2)}{2(x_1^2-x_2^2)^2}\\
			&\qquad\qquad\qquad
			+ \sum_{\substack{ \sigma: \text{ $n$-cycle} \\ \epsilon_2,\cdots,\epsilon_n \in\{\pm 1\}}}
			(-\epsilon_2\cdots\epsilon_n) \cdot
			\prod_{i=1}^n \xi(\epsilon_{\sigma(i)} x_{\sigma(i)}, -\epsilon_{\sigma(i+1)} x_{\sigma(i+1)}),
		\end{split}
		\ee
		where
		\begin{equation*}
			i_{x_1,x_2}
			\frac{x_1x_2(x_2^2+x_1^2)}{2(x_1^2-x_2^2)^2} =
			\sum_{n>0:\text{ odd}} \frac{n}{2} x_1^{-n} x_2^n,
		\end{equation*}
		and $\xi$ is given by:
		\begin{equation*}
			\begin{split}
				\xi(\epsilon_{\sigma(i)} x_{\sigma(i)}, -\epsilon_{\sigma(i+1)} x_{\sigma(i+1)}) =
				\begin{cases}
					\wA (\epsilon_{\sigma(i)} x_{\sigma(i)}, -\epsilon_{\sigma(i+1)} x_{\sigma(i+1)}),
					& \sigma(i)<\sigma(i+1);\\
					-\wA( -\epsilon_{\sigma(i+1)} x_{\sigma(i+1)} ,\epsilon_{\sigma(i)} x_{\sigma(i)}),
					& \sigma(i)>\sigma(i+1),
				\end{cases}
			\end{split}
		\end{equation*}
		and we use the conventions
		$\epsilon_{1} :=1$ and
		$\sigma(n+1):=\sigma(1)$.
	\end{Theorem}
	
	\begin{Remark}
		This theorem is the BKP-analogue of a formula for connected bosonic $n$-point functions
		of a KP tau-function proved by Zhou in \cite{zhou1}.
	\end{Remark}
	
	\begin{Remark}
		The generating series $A(w,x)$ and $\wA(w,x)$ are fermionic $2$-point functions
		associated to this BKP-tau-function.
		See \cite[\S 3]{wy} for details.
	\end{Remark}
	
	The above formula for $n\geq 2$ can be simplified in the following way.
	Notice that for a fixed cycle $\sigma$ and a fixed $j$,
	the variable $z_j$ appears only in two terms
	\begin{equation*}
		\xi (\pm x_i, -\epsilon_j x_j) \cdot \xi (\epsilon_j x_j, \pm x_k)
	\end{equation*}
	in $\prod_{i=1}^{n} \xi(\epsilon_{\sigma(i)} x_{\sigma(i)}, -\epsilon_{\sigma(i+1)} x_{\sigma(i+1)})$
	(where $i$ and $k$ are adjacent to $j$ in this cycle $\sigma$),
	hence replacing $\epsilon_j$ by $-\epsilon_j$ is equivalent to replacing $x_j$ by $-x_j$.
	Therefore,
	replacing $\epsilon_j$ by $-\epsilon_j$ does not change the terms
	with odd orders in $x_j$ in the product
	\begin{equation*}
		\epsilon_2\cdots\epsilon_{n+m}
		\prod_{i=1}^{n+m} \xi(\epsilon_{\sigma(i)} x_{\sigma(i)}, -\epsilon_{\sigma(i+1)} x_{\sigma(i+1)}).
	\end{equation*}
	Moreover,
	the order of $x_j$ in the left-hand side of
	\eqref{eq-npt-general} is always odd,
	and thus we can simply take $\epsilon_2 = \cdots = \epsilon_{n} =1$
	in the right-hand side and then restrict to terms of odd degrees.
	Thus we conclude that:
	\begin{Theorem}
		\label{thm-nptformula-gbgw}
		The connected $n$-point functions associated to the generalized BGW tau-functions
		are given by:
		\be
		\label{eq-npt-viaaffine-gbgw}
		\begin{split}
			&\sum_{i_1,\cdots,i_n> 0: \text{ odd}}
			\frac{\pd^n \log\tau_{\text{BGW}}^{(N)}(\bm t/2)}{\pd t_{i_1}\cdots \pd t_{i_n}} \bigg|_{\bm t=0}
			\cdot x_1^{-i_1}\cdots x_n^{-i_n}
			=
			-\delta_{n,2} \cdot
			\frac{x_1x_2(x_2^2+x_1^2)}{2(x_1^2-x_2^2)^2}\\
			&\qquad\qquad\qquad\qquad\qquad\qquad
			- 2^{n-1}  \cdot \Big[ \sum_{ \sigma: \text{ $n$-cycle} }
			\prod_{i=1}^n \xi( x_{\sigma(i)}, - x_{\sigma(i+1)})
			\Big]_{\text{odd}},
		\end{split}
		\ee
		for $n\geq 1$,
		where $[\cdot]_{\text{odd}}$ means taking the terms of odd degrees in every $x_i$,
		and $\xi$ is:
		\begin{equation*}
			\begin{split}
				\xi( x_{\sigma(i)}, - x_{\sigma(i+1)}) =
				\begin{cases}
					A^{(N)} (x_{\sigma(i)}, - x_{\sigma(i+1)}),
					& \sigma(i) = \sigma(i+1);\\
					\wA^{(N)} ( x_{\sigma(i)}, - x_{\sigma(i+1)}),
					& \sigma(i)<\sigma(i+1);\\
					-\wA^{(N)}( - x_{\sigma(i+1)} , x_{\sigma(i)}),
					& \sigma(i)>\sigma(i+1),
				\end{cases}
			\end{split}
		\end{equation*}
		and $A^{(N)},\wA^{(N)}$ are the generating series of the BKP-affine coordinates:
		\begin{equation*}
			\begin{split}
				&A^{(N)}(w,x) = \sum_{n,m>0} (-1)^{m+n+1}  a_{n,m}^{(N)}\cdot w^{-n} x^{-m}
				- \sum_{n>0} \frac{(-1)^n}{2}  a_{n,0}^{(N)}\cdot (w^{-n}-x^{-n}),\\
				&\wA^{(N)}(w,x) =  A^{(N)}(w,x) -\frac{1}{4} -\half\sum_{i=1}^\infty (-1)^{i} w^{-i} x^i.
			\end{split}
		\end{equation*}
		Here we use the convention
		$\sigma(n+1):=\sigma(1)$.
	\end{Theorem}

	In the next subsection,
	we will give a compact formula for the generating series $\wA^{(N)}(z,w)$,
	which enables us to manipulate the formula \eqref{eq-npt-viaaffine-gbgw}
	more efficiently.

		\begin{Remark}
			In literatures there are other formulas for the connected $n$-point functions of the generalized BGW models
			derived using methods of integrable systems.
			In \cite[Theorem 5]{dyz},
			Dubrovin-Yang-Zagier also represented these connected $n$-point functions as summations over $n$-cycles, 
			with a correction term at $n=2$.
			And in \cite[Theorem 1.1]{br},
			Bertola-Ruzza represent the results as summations over permutations in $S_n$.
		\end{Remark}

	\subsection{Generating series for the BKP-affine coordinates}
	\label{sec-gen-BKPaff}
	
	Now we give a simple formula for $\wA^{(N)}(z,x)$.
	In \cite[\S 5.2]{wy},
	we have proved the following:
	\begin{Proposition}
		[\cite{wy}]
		\label{prop-kdv-genaffine-general}
		Let $\tau(\bm t)$ be a tau-function of the KdV hierarchy,
		and let $\Phi_1(z),\Phi_2(z)$ be the first two basis vectors
		for the corresponding point on the Sato Grassmannian.
		Denote:
		\begin{equation*}
			\Phi_1(z) = 1+\sum_{n\geq 1} a_n z^{-n},
			\qquad
			z^{-1}\Phi_2(z) = 1+\sum_{n\geq 1} b_n z^{-n},
		\end{equation*}
		and let $G(z)$ be the following $2\times 2$ matrix:
		\begin{equation*}
			G(z)= \left[
			\begin{array}{cc}
				1+\sum_{n\geq 1} a_{2n}z^{-n} & \sum_{n\geq 0} b_{2n+1}z^{-n}\\
				\sum_{n\geq 1}a_{2n-1} z^{-n} & 1+\sum_{n\geq 1} b_{2n} z^{-n}
			\end{array}
			\right].
		\end{equation*}
		If $\det G(z)=1$ holds,
		then the generating series \eqref{eq-gen-BKPaffine} of the
		BKP-affine coordinates $\{a_{n,m}\}$ of $\tau(\bm t/2)$ are given by:
		\begin{equation*}
			\begin{split}
				&A(w,z)
				= \frac{w-z+\Phi_1(-z)\Phi_2(-w)-\Phi_1(-w)\Phi_2(-z)}{4(w+z)},\\
				&\wA(w,z)
				= \frac{\Phi_1(-z)\Phi_2(-w)-\Phi_1(-w)\Phi_2(-z)}{4(w+z)}.
			\end{split}
		\end{equation*}
	\end{Proposition}

	Consider the case of generalized BGW tau-functions.
	It is known that the point in the Sato Grassmannian corresponding to
	the KdV tau-function $\tau_\text{gBGW}^{(N)} (\bm t)$ has a basis
	(see Alexandrov \cite[(77)]{al2}):
	\be
	\Big\{ \Phi_j^{(N)}(z)
	= z^{j-1} \Big( 1+
	\sum_{k=1}^\infty \frac{(-\hbar)^k \cdot a_{k}(j-N) }{8^k \cdot k!}
	\cdot z^{-k} \Big) \Big\}_{j= 1}^{\infty},
	\ee
	where
	\begin{equation*}
		a_k(j) = \prod_{i=1}^k \Big( 4(j-1)^2 -(2i-1)^2 \Big).
	\end{equation*}
	(Notice here our notations for $\Phi_j^{(N)} (z)$
	differs from the notations in \cite{al2} by a rescaling $z\mapsto 2z$.)
	In particular,
	the first two basis vectors are:
	\be
	\begin{split}
		\label{eq-KP-first-2basis}
		&\Phi_1^{(N)} (z) = 1+
		\sum_{k=1}^\infty \frac{(-\hbar)^k }{8^k \cdot k!}\cdot
		\prod_{i=1}^k\Big( 4N^2 - (2i-1)^2 \Big)
		\cdot z^{-k},\\
		&\Phi_2^{(N)} (z) = z+
		\sum_{k=1}^\infty \frac{(-\hbar)^k }{8^k \cdot k!}\cdot
		\prod_{i=1}^k\Big( 4(1-N)^2 - (2i-1)^2 \Big)
		\cdot z^{1-k}.
	\end{split}
	\ee
	Now we claim that
	(see \cite[\S 6.4]{wy} for the special case $N=0$):
	\begin{Corollary}
		The generating series $A^{(N)}$ and $\wA^{(N)}$ of the BKP-affine coordinates
		in Theorem \ref{thm-nptformula-gbgw} are given by:
		\be
		\begin{split}
			&A^{(N)}(w,x)
			= \frac{w-x+\Phi_1^{(N)}(-x)\Phi_2^{(N)}(-w)-\Phi_1^{(N)}(-w)\Phi_2^{(N)}(-x)}{4(w+x)},\\
			&\wA^{(N)}(w,x)
			= \frac{\Phi_1^{(N)}(-x)\Phi_2^{(N)}(-w)-\Phi_1^{(N)}(-w)\Phi_2^{(N)}(-x)}{4(w+x)}.
		\end{split}
		\ee
	\end{Corollary}
	\begin{proof}
		By Proposition \ref{prop-kdv-genaffine-general},
		we only need to check $\det G^{(N)}(x) = 1$,
		where:
		\begin{equation*}
			G^{(N)}(x)= \left[
			\begin{array}{cc}
				(\Phi_1^{(N)}(s)+\Phi_1^{(N)}(-s)\big)/2 &
				(\Phi_2^{(N)}(s)+\Phi_2^{(N)}(-s)\big)/2 \\
				(\Phi_1^{(N)}(s)-\Phi_1^{(N)}(-s)\big)/(2s) &
				(\Phi_2^{(N)}(s)-\Phi_2^{(N)}(-s)\big)/(2s)
			\end{array}
			\right],
		\end{equation*}
		and $s=x^\half$.
		Or equivalently,
		we need to check:
		\be
		\Phi_1^{(N)}(-s) \Phi_2^{(N)}(s)
		- \Phi_1^{(N)}(s) \Phi_2^{(N)}(-s) = 2s.
		\ee
		Here $\Phi_1^{(N)}$ and $\Phi_2^{(N)}$ are related by
		(see \cite[(82)]{al2}):
		\be
		\label{eq-KS-phi12}
		\Phi_2^{(N)}(s) = \hbar s \pd_s \Phi_1^{(N)}(s)
		+ s \Phi_1^{(N)}(s) + \hbar(N-\half) \Phi_1^{(N)}(s).
		\ee
		Denote $\Psi(s) = \frac{1}{s}\big(
		\Phi_1^{(N)}(-s) \Phi_2^{(N)}(s)
		- \Phi_1^{(N)}(s) \Phi_2^{(N)}(-s) \big)$,
		then by \eqref{eq-KS-phi12} one has:
		\begin{equation*}
			\Psi(s) = \hbar(\Phi_1^{(N)})'(s) \Phi_1^{(N)}(-s) + \hbar(\Phi_1^{(N)})'(-s) \Phi_1^{(N)}(s)
			+ 2s \Phi_1^{(N)}(s)\Phi_1^{(N)}(-s),
		\end{equation*}
		and then:
		\begin{equation*}
			\frac{d}{ds} \Psi(s)
			= \Phi_1^{(N)}(-s)\cdot f(s) + \Phi_1^{(N)}(s)\cdot f(-s),
		\end{equation*}
		where $f (s) = \hbar (\Phi_1^{(N)})''(s) + 2(\Phi_1^{(N)})'(s)$.
		Using \eqref{eq-KP-first-2basis} one can directly check that:
		\begin{equation*}
			f(s) = \hbar (\Phi_1^{(N)})''(s) + 2(\Phi_1^{(N)})'(s)
			= \frac{\hbar}{s^2} (N^2-\frac{1}{4}) \Phi_1^{(N)}(s),
		\end{equation*}
		and thus  $\frac{d}{ds} \Psi(s)=0$.
		Then $\Psi(s)$ is a constant,
		and it is easy to see that this constant is $2$.
		This completes the proof.
	\end{proof}

	\section{Emergence of Spectral Curve and Its Special Deformation}
	\label{sec-curve}
	
	Now we study the generalized BGW tau-functions using Zhou's formalism of emergent geometry.
	In this section,
	we first give a brief review of this formalism,
	and then construct the spectral curve together with its special deformation.
	They emerges naturally from the free energy of genus zero
	(which can be computed using the Virasoro constraints).
	We also discuss the quantum deformation theory of this spectral curve.
	The emergence of the Eynard-Orantin topological recursion and the quantum spectral curve
	will be addressed in the next two sections.

	\subsection{A brief review of Zhou's emergent geometry}

	First we briefly recall Zhou's idea of emergent geometry.
	This is a formalism which provides the associated geometric structures
	(such as spectral curve, topological recursion, etc.)
	from a given partition function (of a Gromov-Witten type theory).

	The first step is the construction of the spectral curve and its special deformation.
	The notion of the special deformation was first introduced
	by Zhou \cite{zhou9} in the study of topological 2D gravity.
	In general, given a Gromov-Witten type theory (A-theory),
	the spectral curve of this theory together with its special deformation
	(as a B-theory)
	should emerge from the Virasoro constraints.
	More precisely,
	let $F_0(\bm t)$ be the genus zero part of the free energy of a Gromov-Witten type theory,
	then the special deformation is a formal series of the following form
	(with some suitable modification in concrete examples):
	\begin{equation*}
		y = \sum_{n\geq 1}n \tilde t_n x^{n-1}
		+\sum_{n\geq 1} \frac{\pd F_0(\bm t)}{\pd t_n} \cdot  x^{-n-1},
	\end{equation*}
	where $\bm t = (t_1,t_2,t_3,\cdots)$ are the coupling constants,
	and $\tilde{\bm t} =(\tilde t_1,\tilde t_2,\tilde t_3,\cdots)$ differ from $\bm t$
	only by a dilaton shift.
	The spectral curve of this theory is obtained
	from this special deformation by restricting to $\bm t=0$:
	\be
	\label{eq-speccurve-general}
	y = \Big(\sum_{n\geq 1}n \tilde t_n x^{n-1}
	+\sum_{n\geq 1} \frac{\pd F_0(\bm t)}{\pd t_n} \cdot  x^{-n-1}\Big)
	\Big|_{\bm t=0}.
	\ee
	
	Then one can consider the quantum deformation theory of the spectral curve.
	Roughly speaking,
	if a GW type theory is determined by the Virasoro constraints,
	then one hopes that the Virasoro constraints of genus zero can be encoded in the special deformation,
	and the Virasoro constraints at all genera can be encoded
	in a suitable quantization of the special deformation.

	Then one can consider the Eynard-Orantin topological recursion \cite{eo} on the spectral curve.
	In Zhou's formalism,
	the E-O topological recursion on the spectral curve \eqref{eq-speccurve-general}
	is supposed to emerge naturally from the Virasoro constraints.
	Moreover,
	the resulting E-O invariants should coincide with the connected $n$-point functions
	of the original GW type theory.
	It is worth noting that the Bergman kernel for this emergent topological recursion is of the following form:
	\be
	\label{eq-Berg-general}
	B(x_1,x_2) = \Big(\frac{1}{(x_1-x_2)^2} + W_{0,2}(x_1,x_2)\Big) dx_1dx_2,
	\ee
	where $W_{0,2}(x_1,x_2)$ is the connected $2$-point function of genus zero:
	\begin{equation*}
		W_{0,2}(x_1,x_2) = \sum_{n_1,n_2\geq 1}
		\frac{\pd^2 F_0(\bm t)}{\pd t_{n_1} \pd t_{n_1}}\Big|_{\bm t=0}
		\cdot x_1^{-n_1-1} x_2^{-n_2-1},
	\end{equation*}
	and $x$ is exactly the formal variable $x$ in the expression \eqref{eq-speccurve-general},
	regarded as a meromorphic function on this spectral curve.

	Now given a tau-function $\tau(\bm t)$ of the KP hierarchy,
	one can regard it as the partition function of a formal quantum field theory (see \cite[\S 2]{zhou3})
	where the KP-time variables $\bm t=(t_1,t_2,t_3,\cdots)$ are the coupling constants.
	The coefficients of the free energy $\log \tau(\bm t)$ play the role of the connected correlators.
	Then these tau-functions provide a large class of examples of Gromov-Witten type theories.
	See \cite{zhou1, zhou7, zhou8, zhou10, wz, wz2} for the emergent geometry
	for some well-known tau-functions of the KP hierarchy.
	In all these examples,
	the E-O invariants for the emergent spectral curve \eqref{eq-speccurve-general}
	and Bergman kernel \eqref{eq-Berg-general} coincide with the connected $n$-point functions
	associated to the tau-function.
	In \S \ref{sec-eo}, we will see that this ansatz is also true for
	the generalized BGW models.

	\subsection{Special deformation and spectral curve for $\tau_{\text{BGW}}^{(N)}$}
	\label{sec-special}
	
	In this subsection,
	we construct the spectral curve together with its special deformation
	for the generalized BGW tau-functions using Zhou's ansatz.

	In this case,
	the special deformation of the spectral curve is the following formal series in a formal variable $x$:
	\be
	\label{eq-special-def}
	y = \sum_{n\geq 0}(2n+1)(t_{2n+1}-\delta_{n,0}) x^{2n}
	+\sum_{n\geq 0} \frac{\pd \cF_0^{\text{gBGW}}(\bm t)}{\pd t_{2n+1}} \cdot  x^{-2n-2}.
	\ee
	Then the Virasoro constraints of genus zero are encoded in the special deformation
	in the following way:
	\begin{Proposition}
		Virasoro constraints for the free energy $\cF_0^{\text{gBGW}}(\bm t;S)$ of genus zero
		are equivalent to the following condition:
		\be
		(y^2)_- = S^2 \cdot x^{-2},
		\ee
		where we use the notation:
		\begin{equation*}
			\big( \sum_{n\in \bZ} a_n x^n \big)_- = \sum_{n<0} a_n x^n.
		\end{equation*}
	\end{Proposition}
	\begin{proof}
		This can be checked by direct computation.
		By plugging the genus expansion \eqref{eq-F-genusexp} into \eqref{eq-Virasoro-2},
		we obtain the Virasoro constraints at genus zero:
		\begin{equation*}
			\begin{split}
				& -\half \frac{\pd\cF_0^{\text{gBGW}}}{\pd t_1}
				+\half \sum_{k: \text{ odd}} kt_k \frac{\pd \cF_0^{\text{gBGW}}}{\pd t_k}
				-\frac{1}{4}S^2 =0;\\
				& -\half \frac{\pd\cF_0^{\text{gBGW}}}{\pd t_{2n+1}}
				+\frac{1}{4} \sum_{\substack{a+b = 2n\\a,b:\text{ odd}}}
				\frac{\pd \cF_0^{\text{gBGW}}}{\pd t_a} \frac{\pd \cF_0^{\text{gBGW}}}{\pd t_b}
				+ \half \sum_{k: \text{ odd}} kt_k \frac{\pd\cF_0^{\text{gBGW}}}{\pd t_{k+2n}}
				=0,
			\end{split}
		\end{equation*}
		where $n\geq 1$.
		On the other hand, we have:
		\begin{equation*}
			\begin{split}
				&(y^2)_- = \Big( -2 \frac{\pd\cF_0^{\text{gBGW}}}{\pd t_1}
				+2 \sum_{k: \text{ odd}} kt_k \frac{\pd \cF_0^{\text{gBGW}}}{\pd t_k} \Big) z^{-2}\\
				&+ \sum_{n\geq 1}\Big(
				-2 \frac{\pd\cF_0^{\text{gBGW}}}{\pd t_{2n+1}}
				+ \sum_{\substack{a+b = 2n\\a,b:\text{ odd}}}
				\frac{\pd \cF_0^{\text{gBGW}}}{\pd t_a} \frac{\pd \cF_0^{\text{gBGW}}}{\pd t_b}
				+ 2 \sum_{k: \text{ odd}} kt_k \frac{\pd\cF_0^{\text{gBGW}}}{\pd t_{k+2n}}
				\Big) z^{-2n-2},
			\end{split}
		\end{equation*}
		and thus the conclusion is clear.
	\end{proof}

	The spectral curve is supposed to be obtained from the above special deformation
	by taking $\bm t=0$.
	Recall that the connected $1$-point correlators of genus zero can be solved by the Virasoro constraints
	(see \S \ref{sec-rec-npt}):
	\be
	\frac{\pd \cF_0^{\text{gBGW}} (\bm t)}{\pd t_{2n+1}}\Big|_{\bm t=0}
	= \frac{(-1)^{n+1}}{2^{2n+1}} \cdot \frac{1}{n+1}\binom{2n}{n}\cdot S^{2n+2},
	\qquad \forall n\geq 0,
	\ee
	Thus we have:
	\begin{equation*}
		y \big|_{\bm t=0} = -1 + 1-\sqrt{1+\frac{S^2}{x^2}}
		= -\sqrt{1+\frac{S^2}{x^2}}.
	\end{equation*}
	We rewrite this as:
	\be
	\label{eq-speccurve-eo}
	x^2y^2 = x^2 + S^2.
	\ee
	In the framework of emergent geometry,
	the plane curve defined by this equation \eqref{eq-speccurve-eo}
	is the spectral curve associated to the generalized BGW tau-functions.
	In \S \ref{sec-eo}, we will show that the E-O topological recursion on this spectral curve
	emerges naturally from the Virasoro constraints for $\tau_{\text{BGW}}^{(N)}$.

	\subsection{Quantum deformation theory of the spectral curve}
	\label{sec-QDT}
	
	In this subsection,
	we discuss the quantum deformation theory of the above spectral curve.
	We will see that the Virasoro constraints (at all genera)
	can be encoded in terms of a quantization of the special deformation.

	Denote by $\hy(x)$ the following bosonic field on the spectral curve:
	\be
	\hy (x) = \sum_{n\geq 0} \alpha_{-(2n+1)}\cdot x^{2n}
	+ \sum_{n\geq 0} \alpha_{2n+1}\cdot x^{-2n-2},
	\ee
	where $\{\alpha_k\}_{k\in \bZ:\text{ odd}}$ are the free bosons:
	\be
	\alpha_k= \hbar \frac{\pd}{\pd t_k},
	\qquad
	\alpha_{-k}= \hbar^{-1} (kt_k -\delta_{k,1}),
	\qquad \forall k>0: \text{ odd},
	\ee
	which satisfy the canonical commutation relation
	\begin{equation*}
		[\alpha_i, \alpha_j] = i\delta_{i+j,0},
		\qquad \forall i,j\in \bZ: \text{ odd}.
	\end{equation*}
	This bosonic field is a quantization of the special deformation \eqref{eq-special-def}.
	The product of two such fields is:
	\begin{equation*}
		\begin{split}
			\hy(w)\hy(x) =& :\hy(w) \hy(x): + \sum_{n\geq 0} (2n+1)w^{-2n-2}x^{2n}\\
			=& :\hy(w) \hy(x): + \frac{w^2+x^2}{(w^2-x^2)^2},
		\end{split}
	\end{equation*}
	where the normal ordered product $:\alpha_i\alpha_j:$ of two free bosons is given by:
	\begin{equation*}
		:\alpha_i \alpha_j: =\begin{cases}
			\alpha_i \alpha_j, & \text{ if $i\leq j$;}\\
			\alpha_j \alpha_i, & \text{ if $i> j$.}
		\end{cases}
	\end{equation*}
	
	\be
	\label{eq-hy-product}
	\begin{split}
		\hy(x+\epsilon) \hy(x)
		=& :\hy(x+\epsilon) \hy(x): + \frac{2x^2 + 2\epsilon x
			+\epsilon^2}{(2\epsilon x + \epsilon^2)^2}\\
		=& :\hy(x+\epsilon) \hy(x): +
		\frac{1}{2\epsilon^2} + \frac{1}{8x^2} -\frac{\epsilon}{8x^3} + \frac{3\epsilon^2}{32 x^4}
		-\cdots.
	\end{split}
	\ee
	Notice that there is a singular term $\frac{1}{2\epsilon^2}$ in the right-hand side.
	Define the regularized product of two bosonic fields
	by deleting this singular term:
	\be
	\hy(x)^{\odot 2} = \hy(x)\odot \hy(x)
	= \lim_{\epsilon \to 0} \Big( \hy(x+\epsilon) \hy(x) - \frac{1}{2\epsilon^2}\Big),
	\ee
	then from \eqref{eq-hy-product} we see that:
	\be
	\hy(x)^{\odot 2} = :\hy(x)\hy(x): + \frac{1}{8} x^{-2}.
	\ee
	
	\begin{Proposition}
		The Virasoro constraints for the generalized BGW models are equivalent to:
		\be
		\big( \hy(x)^{\odot 2} + (\frac{1}{8} - \hbar^{-2}S^2)x^{-2} \big)_-
		\big(\exp\cF^{\text{gBGW}} (\bm t; S)\big)=0.
		\ee
	\end{Proposition}
	\begin{proof}
		First we have:
		\begin{equation*}
			(\hy(x)^{\odot 2})_-
			= \sum_{n\geq 0} \Big(
			\sum_{\substack{a+b =2n\\a,b\geq 1:\text{ odd}}} \alpha_a\alpha_b
			+ 2 \sum_{\substack{k\geq 1\\k: \text{ odd}}} \alpha_{-k}\alpha_{k+2n}
			\Big) x^{-2n-2} + \frac{1}{8}x^{-2}.
		\end{equation*}
		Now compare this with the Virasoro constraints \eqref{eq-Virasoro-2},
		then we see:
		\begin{equation*}
			\big( \hy(x)^{\odot 2} + (\frac{1}{8} - \hbar^{-2}S^2)x^{-2} \big)_-
			= \sum_{n\geq 0} 4 L_n^{\text{gBGW}} \cdot x^{-2n-2}.
		\end{equation*}
		Thus the conclusion holds.
	\end{proof}

	\section{Emergence of Eynard-Orantin Topological Recursion}
	\label{sec-eo}
	
	In this section,
	we show that the Eynard-Orantin topological recursion for the generalized BGW tau-functions
	emerges naturally from the Virasoro constraints.
	The spectral curve and Bergman kernel for this topological recursion
	are given by \eqref{eq-speccurve-eo} and \eqref{eq-Berg-general} respectively.
	At the end of this section,
	we remark that this topological recursion can be reformulated in terms of another Bergman kernel
	which is understood as the Bergman kernel of type $B$.

	\subsection{Bergman kernel and Eynard-Orantin invariants on the spectral curve}
	\label{sec-bergman-eo}
	
	Let $W_{g,n}^{\text{gBGW}}$ be the connected $n$-point functions defined in \S \ref{sec-rec-npt}.
	
	Now define $B(x_1,x_2)$ to be the following symmetric $2$-differential
	on the spectral curve \eqref{eq-speccurve-eo}:
	\be
	B(x_1,x_2) = \Big( \frac{1}{(x_1-x_2)^2}
	+ W_{0,2}^{\text{gBGW}}(x_1,x_2) \Big)dx_1dx_2.
	\ee
	Or more explicitly,
	\be
	\label{eq-def-Bergman}
	B(x_1,x_2) = \frac{1}{(x_1^2-x_2^2)^2} \Big(
	\frac{x_1^2+x_2^2+2S^2}{\sqrt{(1+\frac{S^2}{x_1^2})(1+\frac{S^2}{x_2^2})}}
	+ 2x_1x_2 \Big) dx_1 dx_2,
	\ee
	by using the explicit expression of $W_{0,2}^{\text{gBGW}}$
	computed in Example \ref{eg-w02}.
	Then $B(x_1,x_2)$ is called a fundamental bidifferential on the spectral curve
	(and is also called a Bergman kernel in the Eynard-Orantin topological recursion \cite{eo}).

	Now one can define a family of symmetric differentials $\{\omega_{g,n}\}_{g\geq 0,n\geq 1}$
	on the spectral curve using the E-O topological recursion,
	for which the input data are the spectral curve \eqref{eq-speccurve-eo},
	the two meromorphic functions $x,y$ on this curve,
	and the Bergman kernel $B(x_1,x_2)$.
	Consider the following parametrization of the spectral curve:
	\be
	\label{eq-param-eo}
	x= \sqrt{z^2-S^2},
	\qquad
	y=\frac{z}{\sqrt{z^2-S^2}},
	\ee
	then $x = x(z)$ has a non-degenerate critical point at $z=0$.
	Near this critical point,
	there is an involution $\sigma(z) = -z$ satisfying $x(z) = x(\sigma(z))$.
	The Eynard-Orantin invariants $\{\omega_{g,n}\}$ are defined by:
	\be
	\begin{split}
		&\omega_{0,1}(z) = (-1)^1 \cdot \big( W_{0,1}^{\text{gBGW}}(x(z)) - 1\big) dx(z)
		= y(z)dx(z),\\
		&\omega_{0,2}(z_1,z_2) = B(x(z_1),x(z_2)) = \frac{1}{(z_1-z_2)^2}dz_1dz_2,
	\end{split}
	\ee
	and for $2g-1+n>0$,
	$\omega_{g,n+1}$ is recursively defined by:
	\be
	\label{eq-eorec}
	\begin{split}
		\omega_{g,n+1}(z_0,z_1,\cdots,z_n)=&
		\Res_{z= 0}K(z_0,z)\bigg[
		\omega_{g-1,n+2}(z,\sigma(z),z_1,\cdots,z_n)\\
		&+\sum_{\substack{g_1+g_2=g\\I\sqcup J=[n]}}^s
		\omega_{g_1,|I|+1}(z,z_I)
		\cdot
		\omega_{g_2,|J|+1}(\sigma(z),z_J)
		\bigg],
	\end{split}
	\ee
	where the recursion kernel $K(z_0,z)$ is defined
	locally near the critical point by:
	\be
	K(z_0,z)=
	\frac{\int_{\sigma(z)}^z B(x(z_0),x(z))}
	{2\big(y(z)-y(\sigma(z))\big)dx(z)},
	\ee
	and $\sum\limits^s$ means excluding all the terms involving $\omega_{0,1}$ in this summation.
	For every $g\geq 0$ and $n\geq 1$,
	the E-O invariant $\omega_{g,n}(z_1,\cdots,z_n)$ is
	a symmetric $n$-differential on the spectral curve \eqref{eq-speccurve-eo}.

	Our main result in this section is the following:
	\begin{Theorem}
		\label{thm-eo-reconstr}
		For every pair $(g,n)$ with $2g-2+n>0$, we have:
		\be
		\begin{split}
			\omega_{g,n}(z_1,\cdots, z_n) =& (-1)^n \cdot
			W_{g,n}^{\text{gBGW}}(x_1,\cdots, x_n) dx_1\cdots dx_n\\
			=& (-1)^n\cdot
			\sum_{\mu_1,\cdots,\mu_n:\text{ odd}} \langle  p_{\mu_1}\cdots p_{\mu_n} \rangle_g^c
			\cdot
			\frac{dx_1\cdots dx_n}{ x_1^{\mu_1+1} \cdots x_n^{\mu_n+1}},
		\end{split}
		\ee
		where $\langle  p_{\mu_1}\cdots p_{\mu_n} \rangle_g^c$
		is the connected correlators of the generalized BGW tau-functions.
	\end{Theorem}
	
	The proof of this theorem will be given in \S \ref{sec-proof-mainthm}.
	\begin{Remark}
		In the special case $N=0$ (and thus $S=0$),
		the spectral curve \eqref{eq-speccurve-eo} degenerates
		and the Eynard-Orantin topological recursion will not work.
		Nevertheless,
		one can first regard $S$ as a formal variable and compute the $n$-point functions $W_{g,n}^{\text{gBGW}}$
		as formal series in $S$ using E-O recursion,
		and then the $n$-point functions for the original BGW model can be obtained by restricting to $S=0$.
	\end{Remark}

		As a straightforward application of the above theorem,
		we give a proof of the following conjecture due to Alexandrov
		(see \cite[(128), (130), and \S 3.4]{al2}) in the rest of this subsection:
		\begin{Conjecture}
			[\cite{al2}]
			\label{conj-al}
			The correlators of the generalized BGW models are given by the Eynard-Orantin topological recursion
			on the following curve:
			\be
			\label{eq-curve-al}
			x^2y^2 -x -\frac{S^2}{4} =0.
			\ee
		\end{Conjecture}
		
		\begin{Remark}
			The original conjecture of Alexandrov is in another form.
			He conjectured that the correlators of the generalized BGW models
			are meromorphic differentials on \eqref{eq-curve-al} whose coefficients are polynomials
			(see \cite[Conjecture 3.7]{al2}),
			and suggested to prove this conjecture using E-O topological recursion.
		\end{Remark}
		
		\begin{Corollary}
			Conjecture \ref{conj-al} holds.
		\end{Corollary}
		\begin{proof}
			To distinguish from the topological recursion on the spectral curve \eqref{eq-speccurve-eo},
			we add a subscript $A$ to the variables $x,y$ in Alexandrov's spectral curve \eqref{eq-curve-al}.
			The curve \eqref{eq-curve-al} admits a parametrization
			(see \cite[(130)]{al2}):
			\begin{equation*}
				x_A =\frac{S^2 \cdot (z_A-1)}{(z_A-2)^2},\qquad
				y_A =\frac{z_A(z_A-2)}{2S\cdot (z_A-1)}.
			\end{equation*}
			Near the critical point is $z_A = 0$,
			an involution $\sigma_A$ is given by:
			\begin{equation*}
				\sigma_A(z_A) = \frac{z_A}{z_A-1}.
			\end{equation*}
			Now compare this parametrization to the parametrization \eqref{eq-param-eo} of the curve \eqref{eq-speccurve-eo},
			and we will see that after taking a M\"obius transformation
			\be
			\label{eqn:z-zA}
			z=\frac{S\cdot z_A}{z_A-2},\qquad
			z_A=\frac{2\, z}{z-S},
			\ee
			the two spectral curves can be related by:
			\begin{equation*}
				x_A (z_A) = \frac{1}{4} x(z)^2,\qquad
				y_A (z_A) = \frac{2y(z)}{x(z)}.
			\end{equation*}
			Then one has:
			\begin{equation*}
				y_A (z_A) dx_A (z_A)
				= \frac{2y(z)}{x(z)} d\Big(\frac{x(z)^2}{4}\Big) = y(z)dx(z) = \omega_{0,1} (z),
			\end{equation*}
			and
			\begin{equation*}
				\begin{split}
					\frac{1}{(z_{A1} - z_{A2})^2} dz_{A1} dz_{A2}
					=& \frac{1}{(\frac{2z_2}{z_1-S} - \frac{2z_2}{z_2-S})^2}
					\cdot \frac{(-2S)dz_1}{(z_1-S)^2} \cdot\frac{(-2S)dz_2}{(z_2-S)^2}\\
					=& \frac{1}{(z_1-z_2)^2} dz_1dz_2
					=\omega_{0,2}(z_1,z_2).
				\end{split}
			\end{equation*}
			(The two equalities above merely mean that the differentials on the left-hand sides
			have the same expressions as the right-hand sides after taking the change of variable \eqref{eqn:z-zA}.
			They are actually differentials on different curves.)
			Moreover,
			one can also check that:
			\begin{equation*}
				\begin{split}
					&\frac{\int_{\sigma_A(z_A)}^{z_A} \frac{1}{(z_{A0}-z_A)^2} dz_{A0}dz_A}
					{2\big(y_A(z_A)-y_A(\sigma_A(z_A))\big)dx_A(z_A)} \\
					=& \frac{(z_A-2)^3 (z_A-1)}{2S\cdot z_A (z_A-z_{A0})(z_Az_{A0}-z_A-z_{A0})} \cdot \frac{dz_{A0}}{dz_A} \\
					=& \frac{z^2 - S^2}{2z(z_0^2 -z^2)} \cdot \frac{dz_0}{dz} = K(z_0,z).
				\end{split}
			\end{equation*}
			Now we see that the E-O topological recursions
			on the two curves \eqref{eq-speccurve-eo} and \eqref{eq-curve-al} have the same initial data
			and recursion kernel,
			and thus they lead to identical $n$-point functions.
			Therefore Conjecture \ref{conj-al} follows from Theorem \ref{thm-eo-reconstr}.
		\end{proof}

	\subsection{Examples of the Eynard-Orantin invariants}
	
	Before proving the above theorem,
	we give some examples of concrete computations of the above topological recursion in this subsection.
	First notice that by the parametrization \eqref{eq-param-eo} we have:
	\begin{equation*}
		dx = \sqrt{1+ \frac{S^2}{x^2}} dz,
	\end{equation*}
	and thus the Bergman kernel \eqref{eq-def-Bergman} can be rewritten as:
	\be
	\begin{split}
		B(x_1,x_2) =&
		\frac{1}{(z_1-z_2)^2}dz_1dz_2,
	\end{split}
	\ee
	and then the recursion kernel is:
	\be
	\label{eq-rec-kernel}
	\begin{split}
		K(z_0,z)=&
		\frac{1}
		{2\big(y(z)-y(-z)\big)dx(z)}
		\int_{-z}^z \frac{1}{(z_0-z)^2}dz_0dz\\
		=& \frac{(z^2 -S^2)}{2z(z_0^2-z^2)}\cdot \frac{dz_0}{dz}.
	\end{split}
	\ee
	
	\begin{Example}
		Now we compute $\omega_{0,3}$.
		By \eqref{eq-eorec} we have:
		\begin{equation*}
			\begin{split}
				\omega_{0,3}(z_0,z_1,z_2)=& \Res_{z= 0}
				K(z_0,z) \big( \omega_{0,2}(z,z_1)\omega_{0,2}(-z,z_2)
				+\omega_{0,2}(z,z_2)\omega_{0,2}(-z,z_1) \big)\\
				=& \frac{S^2}{z_0^2z_1^2z_2^2} dz_0dz_1dz_2,
			\end{split}
		\end{equation*}
		or equivalently,
		\begin{equation*}
			\omega_{0,3}(z_0,z_1,z_2)=
			\frac{S^2 \cdot dx_0 dx_1 dx_2}
			{x_0^2x_1^2x_2^2\cdot
				(1+\frac{S^2}{x_0^2})^{3/2} \cdot (1+\frac{S^2}{x_1^2})^{3/2} \cdot (1+\frac{S^2}{x_2^2})^{3/2}}.
		\end{equation*}
		And then by Theorem \ref{thm-eo-reconstr}:
		\begin{equation*}
			\begin{split}
				&W_{0,3}^{\text{gBGW}}(x_0,x_1,x_2)
				= -\frac{S^2}{x_0^2x_1^2x_2^2}
				+\frac{3S^4}{2 x_0^4x_1^2x_2^2}
				+\frac{3S^4}{2 x_0^2x_1^4x_2^2}
				+\frac{3S^4}{2 x_0^2x_1^2x_2^4}
				-\frac{15S^6}{8 x_0^6x_1^2x_2^2}\\
				& \qquad\qquad\quad
				- \frac{15S^6}{8 x_0^2x_1^6x_2^2}
				-\frac{15S^6}{8 x_0^2x_1^2x_2^6}
				-\frac{9S^6}{4 x_0^4x_1^4x_2^2}
				-\frac{9S^6}{4 x_0^4x_1^2x_2^4}
				-\frac{9S^6}{4 x_0^2x_1^4x_2^4}
				+\cdots.
			\end{split}
		\end{equation*}
	\end{Example}
	
	\begin{Example}
		Now consider $\omega_{0,4}$.
		Using \eqref{eq-eorec} we can compute:
		\begin{equation*}
			\begin{split}
				\omega_{0,4}(z_0,z_1,z_2,z_3)=& \Big(
				\frac{3S^4}{z_0^2 z_1^2 z_2^2 z_3^2}\big(
				\frac{1}{z_0^2}
				+\frac{1}{z_1^2 }
				+\frac{1}{z_2^2 }
				+\frac{1}{z_3^2 }
				\big)-\frac{3S^2}{z_0^2 z_1^2 z_2^2 z_3^2}
				\Big) dz_0dz_2dz_2dz_3,
			\end{split}
		\end{equation*}
		and thus by Theorem \ref{thm-eo-reconstr} we have:
		\begin{equation*}
			\begin{split}
				&W_{0,4}^{\text{gBGW}} (x_0,x_1,x_2,x_3)\\ =&
				\frac{3S^2}{\prod_{i=0}^3 \big( x_i^2 (1+\frac{S^2}{x_i^2})^{3/2} \big)}
				\Big( -1 + \sum_{i=0}^3\frac{S^2}{x_i^2 + S^2} \Big)\\
				=& -\frac{3S^2}{x_0^2x_1^2x_2^2x_3^2}
				+\frac{15S^4}{2 x_0^4x_1^2x_2^2x_3^2}
				+\frac{15S^4}{2 x_0^2x_1^4x_2^2x_3^2}
				+\frac{15S^4}{2 x_0^2x_1^2x_2^4x_3^2}
				+\frac{15S^4}{2 x_0^2x_1^2x_2^2x_3^4}
				+\cdots.
			\end{split}
		\end{equation*}
	\end{Example}
	
	\begin{Example}
		Now we compute $\omega_{1,1}$.
		By \eqref{eq-eorec} we have:
		\begin{equation*}
			\begin{split}
				\omega_{1,1}(z_0)=&
				\Res_{z=0} K(z_0,z) \omega_{0,2}(z,-z)
				= \Big(-\frac{1}{8z_0^2} + \frac{S^2}{8z_0^4}\Big)dz_0,
			\end{split}
		\end{equation*}
		and then by Theorem \ref{thm-eo-reconstr} we have:
		\begin{equation*}
			\begin{split}
				W_{1,1}^{\text{gBGW}} (x_0) =&
				\frac{x_0^3}{8(x_0^2+S^2)^{5/2}}\\
				=& \frac{1}{8 x_0^2} - \frac{5S^2}{16 x_0^4} +
				\frac{35 S^4}{64 x_0^6} - \frac{105 S^6}{128 x_0^8} + \frac{1155 S^8}{1024 x_0^{10}} +\cdots.
			\end{split}
		\end{equation*}
	\end{Example}
	
	\begin{Example}
		Now consider $\omega_{1,2}$.
		By \eqref{eq-eorec} we have:
		\begin{equation*}
			\begin{split}
				\omega_{1,2}(z_0,z_1)=&
				\Res_{z=0} K(z_0,z)\big( \omega_{0,3}(z,-z,z_1) +\\
				&\qquad \omega_{0,2}(z,z_1)\omega_{1,1}(-z)
				+ \omega_{0,2}(-z,z_1)\omega_{1,1}(z)\big)\\
				=& \frac{z_0^4z_1^4-6S^2(z_0^4z_1^2+z_0^2z_1^4)+3S^4 z_0^2z_1^2 +5S^4(z_0^4+z_1^4)}
				{8z_0^6z_1^6} dz_0dz_1.
			\end{split}
		\end{equation*}
		Then by Theorem \ref{thm-eo-reconstr} we have:
		\begin{equation*}
			\begin{split}
				&W_{1,2}^{\text{gBGW}} (x_0,x_1)\\ =&
				\frac{x_0x_1(2 S^8 - 3 S^6 (x_0^2 + x_1^2) - 17 S^4 x_0^2 x_1^2 -
					4 S^2 (x_0^4 x_1^2 + x_0^2 x_1^4) + x_0^4 x_1^4 )}
				{8(x_0^2+S^2)^{7/2}\cdot (x_1^2+S^2)^{7/2}}\\
				=& \frac{1}{8x_0^2x_1^2} -\frac{15S^2}{16x_0^4x_1^2} -\frac{15S^2}{16x_0^2x_1^4}
				+ \frac{175S^4}{64x_0^6 x_1^2} + \frac{175S^4}{64x_0^2 x_1^6}
				+ \frac{93S^4}{32 x_0^4x_1^4} -\cdots.
			\end{split}
		\end{equation*}

	\end{Example}

	\subsection{Proof of Theorem \ref{thm-eo-reconstr}}
	\label{sec-proof-mainthm}
	
	Now we prove Theorem \ref{thm-eo-reconstr}.
	We only need to check that the E-O topological recursion \eqref{eq-eorec}
	is equivalent to the recursion \eqref{eq-rec-npt} derived from the Virasoro constraints.
	
	From the topological recursion \eqref{eq-eorec} one easily sees that when $2g-2+n>0$
	the E-O invariants $\omega_{g,n}$ are of the following form:
	\be
	\omega_{g,n}(z_1,\cdots z_n) = \sum_{k_1,\cdots k_n\geq 0}
	A^{k_1,\cdots,k_n}_{g,n} \cdot \frac{\prod_{i=1}^n(2k_i+1)!!}{z_1^{2k_1+2}\cdots z_n^{2k_n+2}}
	dz_1\cdots dz_n.
	\ee
	Take $x= \sqrt{z^2-S^2}$ and $z=\sqrt{x^2+S^2}$,
	then:
	\be
	z^{-2k} dz = x^{-2k}\cdot \big(1+\frac{S^2}{x^2}\big)^{-k-\half}dx
	= \sum_{m=0}^\infty \binom{-k-\half}{m} S^{2m}x^{-2m-2k} dx,
	\ee
	and for $2g-2+n>0$ one has:
	\be
	\omega_{g,n}(z_1,\cdots z_n) = \sum_{k_1,\cdots k_n\geq 0}
	B^{k_1,\cdots,k_n}_{g,n} \cdot \frac{\prod_{i=1}^n(2k_i+1)!!}{x_1^{2k_1+2}\cdots x_n^{2k_n+2}}
	dx_1\cdots dx_n,
	\ee
	where the numbers $B^{k_1,\cdots,k_n}_{g,n}$ are related to $A^{k_1,\cdots,k_n}_{g,n}$ by:
	\be
	\label{eq-relation-AB}
	\begin{split}
		&B_{g,n}^{l_1,\cdots,l_n} =
		\sum_{k_i + m_i =l_i} \prod_{i=1}^n
		\frac{(-S^2)^{m_i}}{2^{m_i} \cdot m_i!} A_{g,n}^{k_1,\cdots,k_n};\\
		&A_{g,n}^{l_1,\cdots,l_n} =
		\sum_{k_i + m_i =l_i} \prod_{i=1}^n
		\frac{(S^2)^{m_i}}{2^{m_i} \cdot m_i!} B_{g,n}^{k_1,\cdots,k_n}.
	\end{split}
	\ee
	For $(g,n)=(0,1)$ and $(0,2)$,
	we also denote:
	\begin{equation*}
		\begin{split}
			&\omega_{0,1}(z) = dx+\sum_{k\geq 0} B_{0,1}^k\frac{(2k+1)!!}{x^{2k+2}}dx,\\
			&\omega_{0,2}(z_1,z_2) = \frac{1}{(x_1-x_2)^2}
			+ \sum_{k\geq 0} B_{0,2}^{k_1,k_2}\frac{(2k_1+1)!!(2k_2+1)!!}
			{x_1^{2k_1+2}x_2^{2k_2+2}}dx_1dx_2,
		\end{split}
	\end{equation*}
	then direct computation tells that:
	\be
	\label{eq-B0102}
	\begin{split}
		&B_{0,1}^k = - \frac{(-S^2)^{k+1}}{2^{k+1}\cdot (k+1)!\cdot (2k+1)},\\
		&B_{0,2}^{k_1,k_2} = \frac{(-S^2)^{k_1+k_2+1}}{2^{k_1+k_2+1} \cdot k_1!\cdot k_2!\cdot (k_1+k_2+1)}.
	\end{split}
	\ee
	Now Theorem \ref{thm-eo-reconstr} follows from:
	\begin{Lemma}
		For every $g\geq 0$ and $n\geq 1$, we have:
		\be
		\label{eq-lem-pf-eo}
		B_{g,n}^{k_1,\cdots,k_n} = (-1)^n
		\cdot \frac{ \langle p_{2k_1+1} \cdots p_{2k_n+1} \rangle_{g,n}^c}{\prod_{i=1}^n (2k_i+1)!!},
		\qquad \forall k_1,\cdots,k_n\geq 0.
		\ee
	\end{Lemma}
	\begin{proof}
		The unstable cases $(g,n)=(0,1)$ and $(0,2)$
		can be checked directly using the above explicit expressions.
		Now consider the case $(g,n)$ with $2g-2+n>0$.
		The recursion formula \eqref{eq-eorec} is:
		\begin{equation*}
			\begin{split}
				&\omega_{g,n+1}(z_0,z_1,\cdots,z_n)\\
				=&\Res_{z=0}K(z_0,z) \Big(
				\sum_{i=1}^{n}\big(\omega_{0,2}(z,z_i)\omega_{g,n}(-z,z_{[n]\backslash\{i\}})
				+\omega_{g,n}(z,z_{[n]\backslash\{i\}})\omega_{0,2}(-z,z_i)\big)\\
				&\qquad
				+ \omega_{g-1,n+2}(z,\bar z,z_{[n]})
				+\sum_{h=0}^g\sum_{I+J=[n]}^{\text{stable}}
				\omega_{h,|I|+1}(z,z_{I})\omega_{g-h,|J|+1}(-z,z_{J}) \Big),
			\end{split}
		\end{equation*}
		where $\sum\limits^{\text{stable}}$ means we exclude all terms involving
		$\omega_{0,1}$ or $\omega_{0,2}$ in this summation.
		Using the symmetric property of E-O invariants
		(see \cite{eo}):
		\begin{equation*}
			\omega_{g,n}(z_1,\cdots,z_n)= \omega_{g,n}(z_{\sigma(1)},\cdots,z_{\sigma(n)}),
			\qquad \forall \sigma\in S_{n},
		\end{equation*}
		we know that
		$A_{g,n}^{k_1,\cdots,k_n}= A_{g,n}^{k_{\sigma(1)},\cdots,k_{\sigma(n)}}$
		for every $\sigma\in S_{n}$.
		Then we can rewrite the above recursion for $\omega_{g,n}$ as follows:
		\begin{equation*}
			\begin{split}
				&\sum_{{k_0,\cdots, k_n \geq 0}} A_{g,n+1}^{k_0,\cdots,k_n}
				\frac{\prod_{i=0}^{n}(2k_i+1)!!}{z_0^{2k_0+2}\cdots z_n^{2k_n+2}}
				=\Res_{z=0}dz \Big\{\frac{ {S^2-z^2}}{2z(z_0^2-z^2)} \times\\
				&\qquad \sum_{k_1,\cdots,k_n\geq 0}
				\frac{\prod_{i=1}^{n}(2k_i+1)!!}{z_1^{2k_1+2}\cdots z_{n}^{2k_{n}+2}}
				\Big[ 2\sum_{i=1}^{n}\sum_{ a\geq 0}
				\frac{z^{2k_i}}{(2k_i-1)!!}
				\frac{(2a+1)!!}{z^{2a+2}}A_{g,n}^{a,k_{[n]\backslash \{i\}}} +\\
				&\qquad \sum_{a,b\geq 0}\frac{(2a+1)!!(2b+1)!!}{ z^{2a+2b+4}}
				\Bigl(A_{g-1,n+2}^{a,b,k_1,\cdots,k_{n}}
				+\sum_{h=0}^{g}\sum_{I+J=[n]}^{\text{stable}}
				A_{h,|I|+1}^{a,k_{I}} A_{g-h,|J|+1}^{b,k_{J}}\Big) \Big]\Big\}.
			\end{split}
		\end{equation*}
		Notice that near $z=0$ we have:
		\begin{equation*}
			\frac{ {S^2-z^2}}{z_0^2-z^2}=S^2 \cdot\sum_{m=0}^\infty\frac{z^{2m}}{z_0^{2m+2}}
			-\sum_{m=0}^\infty\frac{z^{2m+2}}{z_0^{2m+2}},
		\end{equation*}
		thus by taking residues in the above recursion we obtain
		(after some simplification):
		\begin{equation*}
			\begin{split}
				&\sum_{k_0=0}^{m}\binom{m}{k_0}\cdot
				\Big(-\frac{S^2}{2}\Big)^{m-k_0}
				\cdot \frac{(2k_0+1)!!}{2^{k_0+1}}A_{g,n+1}^{k_0,k_1,\cdots,k_n}\\
				=&-\sum_{i=1}^{n}\sum_{k_0\geq 0} \binom{m+1}{k_0}\cdot
				\Big(-\frac{S^2}{2}\Big)^{m+1-k_0} \cdot
				\frac{(2k_i+2k_0-1)!!\cdot A_{g,n}^{k_i+k_0-1,k_{[n]\backslash\{i\}}}}{2^{k_0}\cdot (2k_i-1)!! }\\
				&-\frac{1}{2}\sum_{a,b\geq 0}\binom{m+1}{a+b+2}
				\frac{(2a+1)!!(2b+1)!!}{2^{a+b+2}}\Big(-\frac{S^2}{2}\Big)^{m-1-a-b}
				\Big(A_{g-1,n+2}^{a,b,k_{[n]}}\\
				&+\sum_{h=0}^g \sum_{I\sqcup J=[n]}^{\text{stable}}
				A_{h,|I|+1}^{a,k_{I}}\cdot A_{g-h,|J|+1}^{b,k_J}\Big).
			\end{split}
		\end{equation*}
		Then we rewrite this recursion as recursion for $B_{g,n}^{k_1,\cdots,k_n}$
		using \eqref{eq-relation-AB}.
		After some simplification,
		the result is:
		\begin{equation*}
			\begin{split}
				&-\sum_{b=0}^{m}\frac{(2m-2b-1)!! \cdot(2b+1)!! \cdot
					(-\frac{S^2}{2})^{m-b}}{2^{m+1} \cdot (m-b)! \cdot (2m-2b-1)} \cdot B_{g,n+1}^{b,k_1,\cdots,k_n}\\
				=&-\sum_{i=1}^{n} \Bigl(\frac{(2m+2k_i+1)!!}{(2k_i-1)!! \cdot 2^{m+1}}
				\cdot B_{g,n}^{m+k_i,k_{[n]\backslash\{i\}}}\\
				&\qquad +\sum_{a+b=m-1}
				\frac{(2a+1)!!\cdot (2b+1)!! \cdot (-\frac{S^2}{2})^{a+k_i+1}}
				{2^{m+1} \cdot a! \cdot  k_i! \cdot(a+k_i+1)} \cdot
				B_{g,n}^{b,k_{[n]\backslash\{i\}}}\Bigr) \\
				&-\frac{1}{2}\sum_{a'+b'=m-1}\frac{(2a'+1)!!(2b'+1)!!}{2^{m+1}}
				\Bigl(B_{g-1,n+2}^{a',b',k_{[n]}}
				+\sum_{\substack{g_1+g_2=g\\I+J=[n]}}^{\text{stable}}
				B_{g_1,|I|+1}^{a',k_{I}}B_{g_2,|J|+1}^{b',k_J}\Bigr).
			\end{split}
		\end{equation*}
		Now using \eqref{eq-B0102} we rewrite this recursion as:
		\begin{equation*}
			\begin{split}
				-\sum_{b=0}^{m}
				\tB_{0,1}^{m-b-1} &\tB_{g,n+1}^{b,k_1,\cdots,k_n}
				=\sum_{i=1}^{n} \Big( \tB_{g,n}^{m+k_i,k_{[n]\backslash\{i\}}}
				+\sum_{a+b=m-1}
				\tB_{0,2}^{a,k_i} \tB_{g,n}^{b,k_{[n]\backslash\{i\}}}\Big) \\
				&+\frac{1}{2}\sum_{a'+b'=m-1}
				\Big( \tB_{g-1,n+2}^{a',b',k_{[n]}}
				+\sum_{\substack{g_1+g_2=g\\I+J=[n]}}^{\text{stable}}
				\tB_{g_1,|I|+1}^{a',k_{I}}\tB_{g_2,|J|+1}^{b',k_J}\Big).
			\end{split}
		\end{equation*}
		where we use the notation $\tB_{0,1}^{-1} = -1$, and:
		\begin{equation*}
			\tB_{g,n}^{k_1,\cdots,k_n} = (-1)^n\cdot
			\prod_{i=1}^n (2k_i+1)!!\cdot B_{g,n}^{k_1,\cdots,k_n},
			\qquad \forall k_1,\cdots,k_n\geq 0.
		\end{equation*}
		This matches with the Virasoro constraints \eqref{eq-rec-corr} for the connected correlators
		(notice that the summation in the right-hand side of \eqref{eq-rec-corr} involves
		unstable terms $\langle p_i \rangle_0^c$ and $\langle p_i p_j \rangle_0^c$).
		Thus we conclude that
		\begin{equation*}
			\tB_{g,n}^{k_1,\cdots,k_n} =
			\langle p_{2k_1+1} \cdots p_{2k_n+1} \rangle_{g,n}^c,
			\qquad \forall k_1,\cdots,k_n\geq 0,
		\end{equation*}
		since the two sides are determined by the same recursion formula and initial values.
		This completes the proof.
	\end{proof}

	\subsection{A remark on Bergman kernel of type $B$}
	\label{sec-berg-B}
	
	In this subsection,
	we give a remark on the Bergman kernel of type $B$
	which provides a reformulation of the above topological recursion.

	Let $\cC$ be a rational spectral curve with a complex parameter $z$,
	and consider the following symmetric $2$-differential on the spectral curve:
	\be
	\label{eq-Berg-B}
	B(z_1,z_2) =
	\frac{z_1^2 + z_2^2}{(z_1^2-z_2^2)^2}dz_1dz_2
	= \half \Big(
	\frac{1}{(z_1-z_2)^2} + \frac{1}{(z_1+z_2)^2}
	\Big) dz_1dz_2 .
	\ee
	In literatures the above $2$-differential is supposed to serve as the Bergman kernel
	for the topological recursion associated to a tau-function of the BKP hierarchy.
	The first example of such a Bergman kernel in topological recursion
	might be the Witten-Kontsevich tau-function \cite{zhou2},
	see \cite[(26)]{zhou2}.
	In that work,
	Zhou proved that the Virasoro constraints \cite{dvv, fkn} for
	the Witten-Kontsevich tau-function $\tau_{\text{WK}}(\bm t)$
	are equivalent the Eynard-Orantin topological recursion on the Airy curve,
	and Bergman kernel he chose is exactly \eqref{eq-Berg-B}.
	Zhou treated $\tau_{\text{WK}}(\bm t)$ as a tau-function of the KdV hierarchy
	(due to the Witten Conjecture/Kontsevich Theorem \cite{wit, kon}),
	and now we may expect that the emergence of such a Bergman kernel in this case
	is because $\tau_{\text{WK}}(\bm t/2)$ is a tau-function of the BKP hierarchy.
	The function $(x_1^2 + x_2^2)/(x_1^2-x_2^2)^2$ also appeared in Bertola-Dubrovin-Yang \cite[Theorem 1.7]{bdy},
		as the correction term at $n=2$ in a formula for the $n$-point functions of an arbitrary tau-function of KdV hierarchy.
	Recently in \cite{as},
	Alexandrov and Shadrin proved the blobbed topological recursion
	for a class of hypergeometric $2$-BKP tau-functions
	where the Bergman kernel is \eqref{eq-Berg-B} (see \cite[(7.1)]{as}),
	and their results confirmed the conjectural topological recursion \cite[Conjecture 1.3]{gkl}
	for the spin Hurwitz numbers~\cite{eop} with completed cycles.
	It is worth mentioning that the function $(x_1^2 +x_2^2)/(x_1^2-x_2^2)^2$ in this Bergman kernel
	differs from the correction term at $n=2$
	in the right-hand side of \eqref{eq-npt-general}
	only by a factor $\frac{x_1x_2}{2}$.

	Now in this subsection,
	we make a simple observation that the above Eynard-Orantin topological recursion
	for generalized BGW models can also be reformulated using such a Bergman kernel.
	Consider the spectral curve \eqref{eq-speccurve-eo} with parametrization \eqref{eq-param-eo}.
	Similar to the construction \eqref{eq-Berg-general},
	we consider the following Bergman kernel modified by all connected correlators
	of type $(g,n) = (0,2)$:
	\be
	\tB(x_1,x_2) = \Big( \half\big(\frac{1}{(x_1-x_2)^2} +\frac{1}{(x_1+x_2)^2}\big)
	+ W_{0,2}^{\text{gBGW}}(x_1,x_2) \Big)dx_1dx_2.
	\ee
	Or more explicitly (by using the result in Example \ref{eg-w02}):
	\be
	\tB(x_1,x_2) = \frac{1}{(x_1^2-x_2^2)^2} \cdot
	\frac{x_1^2+x_2^2+2S^2}{\sqrt{(1+\frac{S^2}{x_1^2})(1+\frac{S^2}{x_2^2})}} dx_1 dx_2.
	\ee
	Using the parametrization \eqref{eq-param-eo},
	we can rewrite it as:
	\be
	\tB(x_1,x_2) =
	\half \Big( \frac{1}{(z_1-z_2)^2} + \frac{1}{(z_1+z_2)^2} \Big)dz_1dz_2,
	\ee
	where $x_i = x(z_i)$.
	Then we have:
	\begin{Theorem}
		Define a family of symmetric differentials $\{\tomega_{g,n}\}_{g\geq 0,n\geq 1}$
		on the spectral curve \eqref{eq-speccurve-eo} as follows:
		\begin{equation*}
			\begin{split}
				&\tomega_{0,1}(z) = y(z)dx(z),
				\qquad
				\tomega_{0,2}(z_1,z_2) = \tB(x(z_1),x(z_2)),\\
				&\tomega_{1,1}(z_0)
				= \Big(-\frac{1}{8z_0^2} + \frac{S^2}{8z_0^4}\Big)dz_0,
			\end{split}
		\end{equation*}
		and for $2g-1+n>0$ and $(g,n)\not = (1,0)$,
		the differential $\omega_{g,n+1}$ is recursively defined by the Eynard-Orantin topological recursion:
		\be
		\label{eq-eorec-B}
		\begin{split}
			\tomega_{g,n+1}(z_0,z_1,\cdots,z_n)=&
			\Res_{z= 0}\tK(z_0,z)\bigg[
			\tomega_{g-1,n+2}(z,-z,z_1,\cdots,z_n)\\
			&+\sum_{\substack{g_1+g_2=g\\I\sqcup J=[n]}}^s
			\tomega_{g_1,|I|+1}(z,z_I)
			\tomega_{g_2,|J|+1}(-z,z_J)
			\bigg],
		\end{split}
		\ee
		where the recursion kernel $\tK(z_0,z)$ is defined by:
		\be
		\tK(z_0,z)=
		\frac{\int_{\sigma(z)}^z \tB(x(z_0),x(z))}
		{2\big(y(z)-y(\sigma(z))\big)dx(z)}.
		\ee
		Then for every pair $(g,n)$ with $2g-2+n>0$, we have:
		\be
		\begin{split}
			\tomega_{g,n}(z_1,\cdots, z_n) =& (-1)^n \cdot
			W_{g,n}^{\text{gBGW}}(x_1,\cdots, x_n) dx_1\cdots dx_n.
		\end{split}
		\ee
	\end{Theorem}
	\begin{Remark}
		Here we exclude the case $\tomega_{1,1}$ in the recursive definition \eqref{eq-eorec-B}
		since in this case $\tomega_{0,2}(z,-z)$ in the right-hand side is not well-defined.
	\end{Remark}
	\begin{proof}
		It suffices to check that $\tomega_{g,n} = \omega_{g,n}$ for $2g-2+n>0$.
		We prove by induction.
		Assume that $\tomega_{g',n'} = \omega_{g',n'}$ holds for all $(g',n')$ (where $2g'-2+n'>0$)
		satisfying $g'<g$, or $g'=g$ and $n'<n+1$.
		We need to prove:
		\begin{equation*}
			\begin{split}
				&\Res_{z= 0} \tK(z_0,z)\bigg[
				\tomega_{g-1,n+2}(z,-z,z_{[n]})
				+\sum_{\substack{g_1+g_2=g\\I\sqcup J=[n]}}^s
				\tomega_{g_1,|I|+1}(z,z_I)
				\tomega_{g_2,|J|+1}(-z,z_J)
				\bigg]\\
				=&
				\Res_{z= 0}K(z_0,z)\bigg[
				\omega_{g-1,n+2}(z,-z,z_{[n]})
				+\sum_{\substack{g_1+g_2=g\\I\sqcup J=[n]}}^s
				\omega_{g_1,|I|+1}(z,z_I)
				\omega_{g_2,|J|+1}(-z,z_J)
				\bigg].
			\end{split}
		\end{equation*}
		First notice that the recursion kernel
		\begin{equation*}
			\begin{split}
				\tK(z_0,z)=&
				\frac{1}
				{4\big(y(z)-y(-z)\big)dx(z)}
				\int_{-z}^z \big( \frac{1}{(z_0-z)^2} + \frac{1}{(z_0+z)^2} \big)dz_0dz\\
				=& \frac{(z^2 -S^2)}{2z(z_0^2-z^2)}\cdot \frac{dz_0}{dz}
			\end{split}
		\end{equation*}
		coincide with the recursion kernel $K(z_0,z)$ (see \eqref{eq-rec-kernel}) for $\omega_{g,n}$.
		Moreover,
		\begin{equation*}
			\tK(z_0,z)= K(z_0,z)
			= -\half
			\Big( \sum_{m\geq 0} S^2\cdot\frac{z^{2m-1}}{z_0^{2m+2}}
			-\sum_{m\geq 0} \frac{z^{2m+1}}{z_0^{2m+2}} \Big)
		\end{equation*}
		is a Laurent series in $z$ which contains only terms of odd degrees in $z$,
		thus it suffices to prove that:
		\be
		\label{eq-pf-B-eo}
		\begin{split}
			&\bigg[
			\tomega_{g-1,n+2}(z,-z,z_{[n]})
			+\sum_{\substack{g_1+g_2=g\\I\sqcup J=[n]}}^s
			\tomega_{g_1,|I|+1}(z,z_I)
			\tomega_{g_2,|J|+1}(-z,z_J)
			\bigg]_{\text{even}}\\
			=&
			\bigg[
			\omega_{g-1,n+2}(z,-z,z_{[n]})
			+\sum_{\substack{g_1+g_2=g\\I\sqcup J=[n]}}^s
			\omega_{g_1,|I|+1}(z,z_I)
			\omega_{g_2,|J|+1}(-z,z_J)
			\bigg]_{\text{even}},
		\end{split}
		\ee
		where $[\cdot]_{\text{even}}$ means taking the terms of even degrees in $z$.
		Notice that we have:
		\begin{equation*}
			[\omega_{0,2}]_{\text{even}} = [\tomega_{0,2}]_{\text{even}},
		\end{equation*}
		and by induction hypothesis and Theorem \ref{thm-eo-reconstr} we know that
		$\omega_{g',n'}= \tomega_{g',n'}$
		contains only terms of even degrees for $g'<g$, or $g'=g$ and $n'<n+1$
		(where $2g'-2+n'>0$),
		and thus the identity \eqref{eq-pf-B-eo} is clear.
	\end{proof}

	\section{Emergence of Quantum Spectral Curve of Type $B$}
	\label{sec-qsc-B}
	
	In this section,
	we discuss the emergence of the quantum spectral curve
	of the classical curve \eqref{eq-speccurve-eo}.
	We first show that there exists a Kac-Schwarz operator $P$ of type $B$
	for the generalized BGW tau-function which annihilates the BKP-wave function
	using BKP-affine coordinates.
	(This operator has been discussed by Alexandrov in \cite{al3}.)
	Then we show that the semi-classical limit of $P$ gives the classical curve \eqref{eq-speccurve-eo}.
	Since we already know that the generalized BGW models can be reconstructed from
	the Eynard-Orantin topological recursion on \eqref{eq-speccurve-eo},
	this operator $P$ provides a quantum spectral curve in the sense of Gukov-Su{\l}kowski.

	\subsection{Emergence of Kac-Schwarz operators from BKP-affine coordinates}
	
	In this subsection we recall the formulation of Kac-Schwarz operators of type $B$
	in terms of BKP-affine coordinates,
	see \cite[\S 3]{jwy} for details.
	
	Let $\tau(\bm t)$ be a tau-function of the BKP hierarchy satisfying the condition $\tau(0) =1$,
	where $\bm t = (t_1,t_3,t_5,\cdots)$ are the BKP-time variables.
	The wave function associated to $\tau(\bm t)$ is defined by Sato's formula \cite{djkm}:
	\be
	\label{eq-SatoFormula}
	w_B (\bm t;z) = \exp\Big(\sum_{k=0}^\infty t_{2k+1} z^{2k+1}\Big)\cdot
	\frac{\tau (t_1 - \frac{2}{z}, t_3 -\frac{2}{3z^3}, t_5-\frac{2}{5z^5},\cdots)}{\tau (t_1,t_3,t_5,\cdots)}.
	\ee
	In particular,
	$w_B(0;z)$ is a principal specialization of $\tau$:
	\begin{equation*}
		w_B (0;z) =
		\tau \big(- \frac{2}{z},  -\frac{2}{3z^3}, -\frac{2}{5z^5},\cdots \big).
	\end{equation*}
	One can associate a linear subspace $U_\tau$ of $\cH = \bC[z] \oplus z^{-1}\bC[[z^{-1}]]$
	to $\tau(\bm t)$ in the following way:
	\be
	U_\tau =
	\text{span} \big\{\pd_{t_1}^kw_B(0;z) \big\}_{k\geq 0}.
	\ee
	An operator $P$ (acting on formal Laurent series in $z$) is called
	a Kac-Schwarz operator of type $B$ for the tau-function $\tau(\bm t)$,
	if it satisfies:
	\be
	P(U_\tau) \subset U_\tau.
	\ee
	The subspace $U_\tau \subset \cH$ emerges naturally from the tau-function $\tau(\bm t)$ via
	its BKP-affine coordinates $\{a_{n,m}\}_{n,m\geq 0}$,
	due to the following:
	\begin{Theorem}
		[\cite{jwy}]
		\label{thm-wave-affine}
		We have:
		\be
		U_\tau =
		\text{span} \big\{ z^k + \sum_{i=1}^\infty 2(-1)^{i} (a_{k,i} - a_{k,0}a_{0,i})
		z^{-i} \big\}_{k\geq 0}.
		\ee
		In particular,
		the first basis vector coincides with $w_B(0;z)$:
		\be
		w_B(0;z) =1+ \sum_{i=1}^\infty 2(-1)^i\cdot  a_{0,i} z^{-i}.
		\ee
	\end{Theorem}
	\begin{Remark}
		The above linear subspace $U_\tau \subset \cH$ and
		the BKP tau-function $\tau(\bm t)$ are determined by each other,
		since the BKP-affine coordinates $\{a_{n,m}\}_{n,m\geq 0}$ satisfy
		the anti-symmetry condition $a_{n,m} = -a_{m,n}$.
	\end{Remark}
	
	\begin{Remark}
		The subspace $U_\tau \subset \cH$ can also be spanned by
		the fermionic $1$-point functions associated to $\tau(\bm t)$,
		see \cite[\S 3.2]{jwy} for details.
	\end{Remark}

	In what follows,
	we will denote by $\Phi_k^B(z)$ the $k$-th basis vector for $U_\tau$:
	\be
	\Phi_k^B(z) =
	z^k + \sum_{i=1}^\infty 2(-1)^{i} ( a_{k,i} -a_{k,0} a_{0,i} )
	z^{-i},
	\qquad k\geq 0.
	\ee
	Then $U_\tau = \text{span}\{ \Phi_k^B(z)\}_{k=0,1,2,\cdots}$.
	Similar to the case of KP hierarchy and the usual Sato Grassmannian,
	in general we are interested in those BKP tau-functions for which
	there are two Kac-Schwarz operators $(P,Q)$ of type $B$,
	such that:
	\be
	\label{eq-KS-conditionPQ}
	\begin{split}
		& P(\Phi_0^B) = 0;\\
		& Q(\Phi_k^B) - c_{k+1}\cdot \Phi_{k+1}^B \in \text{span}\{\Phi_0^B,\Phi_1^B,\cdots,\Phi_{k}^B\}.
	\end{split}
	\ee
	where $\{c_{k+1}\}$ are some nonzero constants.
	Moreover,
	we hope that they satisfy the canonical commutation relation:
	\be
	[P,Q]=1.
	\ee
	If the equation $P(\Phi_0^B)=0$ has a unique solution
	$\Phi_0^B \in 1+ z^{-1}\bC[[z^{-1}]]$,
	then $\{a_{n,m}\}$ and hence $\tau(\bm t)$ are uniquely determined,
	since $\Phi_k^B(z)$ for every $k>0$ can be recursively computed by applying $Q$
	to $\Phi_{k-1}^B(z)$.

	\subsection{Kac-Schwarz operators of type $B$ for generalized BGW models}
	
	In last subsection we have reviewed the formulation of
	Kac-Schwarz operators of type $B$ in terms of BKP-affine coordinates.
	From now on,
	we apply Theorem \ref{thm-wave-affine} and the explicit expressions of
	the BKP-affine coordinates given in \S \ref{sec-affinecoord-gbgw} to
	derive the Kac-Schwarz operators of type $B$ for the generalized BGW tau-functions.

	Recall that the BKP-affine coordinates of $\tau_{\text{BGW}}^{(N)}(\bm t/2)$
	are given by \eqref{eq-BKPaffinecor-gbgw}.
	Thus for each $N \in \bC$,
	the linear subspace $U_{\tau_{\text{BGW}}^{(N)}(\bm t/2)} \subset \cH$
	is spanned by the following basis vectors $\{\Phi_k^{B,(N)}(z)\}_{k\geq 0}$:
	\be
	\begin{split}
		\Phi_k^{(N),B}(z) =&
		z^k + \sum_{i=1}^\infty 2(-1)^{i} ( a_{k,i}^{(N)} -a_{k,0}^{(N)} a_{0,i}^{(N)} )
		z^{-i}\\
		=& z^k + \sum_{i=0}^\infty
		(-1)^i \cdot \frac{i}{i+k}\cdot \frac{\hbar^{k+i}}{2^{3k+3i}\cdot k!\cdot i!}
		\prod_{a=1}^k \theta(a)\prod_{b=1}^i \theta(b).
	\end{split}
	\ee
	where $\theta$ is the function:
	\begin{equation*}
		\theta (z) = (2z-1)^2 -4N^2,
	\end{equation*}
	and we use the convention $\prod_{a=1}^0 \theta(a) =1$.
	Then we have:
	\begin{Proposition}
		Let $P$ and $Q$ be the following operators:
		\be
		\label{eq-def-PQ}
		\begin{split}
			&P=\hbar^3 \Big( (z\pd_z+\half )^2 -N^2\Big)
			\Big( \pd_z -\frac{\hbar}{2z^2} \big(
			(z\pd_z-\frac{1}{2})^2 -N^2
			\big) \Big),\\
			&Q=\hbar^{-2} \big( (z\pd_z-\half)^2 -N^2 \big)^{-1} z.
		\end{split}
		\ee
		Then one has $[P,Q]=\hbar$, and:
		\begin{equation*}
			\begin{split}
				& P(\Phi_k^{(N),B}) =
				\frac{\hbar^3}{4} \theta(k) \cdot
				\Big( k \Phi_{k-1}^{(N),B}
				-\frac{\hbar}{8}  \theta(k-1) \Phi_{k-2}^{(N),B}
				\Big),\\
				& Q(\Phi_k^{(N),B}) = \frac{\hbar^{-2}}{4}
				\theta(k+1)^{-1}  \Phi_{k+1}^{(N),B}
				- \frac{\hbar^{k-1} \cdot \theta(1)\cdot \prod_{j=1}^k\theta(j)}
				{2^{3k+3}\cdot (k+1)!\cdot (\frac{1}{4}-N^2)} \Phi_{0}^{(N),B},
			\end{split}
		\end{equation*}
		for every $k\geq 0$,
		where we use the conventions $\Phi_{-1}^{(N),B} = \Phi_{-2}^{(N),B} =0$.
	\end{Proposition}
	\begin{proof}
		Notice that:
		\begin{equation*}
			\begin{split}
				&P(z^k)= \frac{\hbar^3}{4} \theta(k)
				\Big( kz^{k-1} -\frac{\hbar}{8} \theta(k-1) z^{k-2} \Big),\\
				&Q(z^k)= \frac{\hbar^{-2}}{4}\theta(k+1) \cdot z^{k+1}.
			\end{split}
		\end{equation*}
		Then the conclusion can be checked directly using these two identities.
	\end{proof}
	
	This proposition tells that the operators $(P,Q)$ are Kac-Schwarz operators
	of type $B$ for the tau-functions $\tau_{\text{BGW}}^{(N)}(\bm t/2)$.
	
	\subsection{Quantum spectral curve of type $B$ for $\tau_{\text{BGW}}^{(N)}({\bf t}/2)$}
	
	Now we explain that the Kac-Schwarz operator $P$ defined by \eqref{eq-def-PQ} gives
	the quantum spectral curve (in the sense of Gukov-Su{\l}kowski) of
	the plane curve:
	\be
	\label{eq-speccurve-eo-qsc}
	x^2y^2 = x^2 + S^2.
	\ee
	
	We have already seen that $P$ annihilates the BKP-wave function $w_B^{(N)}(\bm t;z)$
	associated to the tau-function $\tau_{\text{BGW}}^{(N)}(\bm t)$
	evaluated at $\bm t=0$:
	\begin{equation*}
		P \big(w_B^{(N)}(0;z)\big) =
		P\big( \Phi_0^{(N),B}(z) \big)  =0,
	\end{equation*}
	where $w_B^{(N)}(0;z)$ coincides with the following principal specialization of
	the tau-function due to Sato's formula \eqref{eq-SatoFormula}:
	\begin{equation*}
		w_B^{(N)}(0;z) =
		\tau_{\text{BGW}}^{(N)}\big(
		-\frac{2}{z}, -\frac{2}{3z^3}, -\frac{2}{5z^5},\cdots \big).
	\end{equation*}
	Moreover,
	in \S \ref{sec-eo} we have shown that the free energy $\log \tau_{\text{BGW}}^{(N)}$
	can be reconstructed from the Eynard-Orantin topological recursion
	on the spectral curve \eqref{eq-speccurve-eo-qsc},
	therefore we only need to check that \eqref{eq-speccurve-eo-qsc} is the semi-classical limit of $P$.
	Denote:
	\be
	\hat x= z\cdot, \qquad\qquad \hat y=\hbar \pd_z,
	\ee
	then $[y,x] = \hbar$, and the operator $P$ can be rewritten as follows:
	\be
	\begin{split}
		P=& \Big( (\hbar z\pd_z+\frac{\hbar}{2} )^2 -S^2\Big)
		\Big( \hbar\pd_z -\frac{1}{2z^2} \big(
		(\hbar z\pd_z-\frac{\hbar }{2})^2 -S^2
		\big) \Big) \\
		=& \Big( (\hat x \hat y+\frac{\hbar}{2} )^2 -S^2\Big)
		\Big( \hat y - \half \hat x^{-2} \big(
		(\hat x \hat y-\frac{\hbar }{2})^2 -S^2
		\big) \Big),
	\end{split}
	\ee
	where $S = \hbar\cdot N$.
	By taking the semi-classical limit,
	we obtain the following function $H(x,y)$
	on the $(x,y)$-plane:
	\be
	H(x,y) = (x^2y^2-S^2) \Big(y-\frac{1}{2x^2}(x^2y^2-S^2)\Big).
	\ee
	Notice that the spectral curve defined by $H(x,y)$ is reducible,
	and thus we are not dealing with the standard case of the topological recursion.
	Nevertheless,
	here we take the second factor in this function which defines a plane curve:
	\begin{equation*}
		2x^2 y = x^2y^2-S^2.
	\end{equation*}
	This classical curve coincides with the spectral curve \eqref{eq-speccurve-eo-qsc}
	for the Eynard-Orantin topological recursion after a shift $y\mapsto y+1$,
	and thus we will regard the Schr\"odinger equation $P\big( \Phi_0^{(N),B}(z) \big) =0$
	as its quantum spectral curve in the sense of Gukov-Su{\l}kowski \cite{gs}.

	\vspace{.2in}
	
	{\em Acknowledgements}.
	We thank the anonymous referees for helpful suggestions.
	We also thank Prof. Jian Zhou for pointing out that
	the correction term in the formula \eqref{eq-npt-viaaffine-gbgw}
	coincide with the Bergman kernel used in \cite{zhou2},
	and thank Prof. Huijun Fan, Prof. Shuai Guo, Prof. Xiaobo Liu, Prof. Xiangyu Zhou for encouragement,
	and Dr. Ce Ji for helpful discussions.
	C. Yang  is supported by the NSFC grants (No. 12288201, 12401079),
the China Postdoctoral Science Foundation (No. 2023M743717),
and the China National Postdoctoral Program for Innovative Talents (No. BX20240407).


\begin{thebibliography}{9}
		%\addcontentsline{toc}{chapter}{Bibliography}
		
		
		\bibitem{al2}Alexandrov A. Cut-and-join description of generalized Br\'ezin-Gross-Witten model. Advances in Theoretical and Mathematical Physics, 2018, 22(6).
		
		\bibitem{al5}Alexandrov A. KdV solves BKP. Proceedings of the National Academy of Sciences, 2021, 118(25):e2101917118.
		
		\bibitem{al3}Alexandrov A. Generalized Br\'ezin-Gross-Witten tau-function as a hypergeometric solution of the BKP hierarchy. Advances in Mathematics, 2023, 412:108809.
		
		\bibitem{al4}Alexandrov A. Intersection numbers on $\Mbar_{g,n}$ and BKP hierarchy. Journal of High Energy Physics, 2021, 2021(9):013.
		
		\bibitem{as}Alexandrov A, Shadrin A. Elements of spin Hurwitz theory: closed algebraic formulas, blobbed topological recursion, and a proof of the Giacchetto-Kramer-Lewa\'nski conjecture. Selecta Math. (N.S.), 2023, 29: 26.
		
		\bibitem{an}Anderson P W. More is different. Science, 1972, 177(4047):393-396.
		
		\bibitem{by}Balogh F, Yang D. Geometric interpretation of Zhou's explicit formula for the Witten-Kontsevich tau function. Letters in Mathematical Physics, 2017, 107(10):1837-1857.
		
		\bibitem{bdy} Bertola M, Dubrovin B, Yang D. Correlation functions of the KdV hierarchy and applications to intersection numbers over $\Mbar_ {g, n}$. Physica D: Nonlinear Phenomena. 2016, 327:30-57.
		
		\bibitem{br}Bertola M, Ruzza G. Br\'ezin-Gross-Witten tau function and isomonodromic deformations. Commun, Number Theory Phys, 2019, 13:827-883.
		
		\bibitem{bg}Br\'ezin E, Gross D J. The external field problem in the large $N$ limit of QCD. Physics Letters B, 1980, 97(1):120-124.
		
		\bibitem{cgg}Chidambaram N K, Garcia-Failde E, Giacchetto A. Relations on $\Mbar_{g,n}$ and the negative $r$-spin Witten conjecture. arXiv preprint arXiv:2205.15621, 2022..
		
		\bibitem{djkm}Date E, Jimbo M, Kashiwara M, Miwa T. Transformation groups for soliton equations IV. A new hierarchy of soliton equations of KP-type. Physica D, 1982, 4(3):343-365.
		
		\bibitem{djm}Date E, Jimbo M, Miwa T. Solitons: Differential equations, symmetries and infinite dimensional algebras. Cambridge University Press, 2000.
		
		\bibitem{dm}Deligne P, Mumford D. The irreducibility of the space of curves of given genus. Publications Math\'ematiques de l'IH\'ES, 1969, 36(1): 75-109.
		
		\bibitem{dvv}Dijkgraaf R, Verlinde H, Verlinde E. Loop equations and Virasoro constraints in nonperturbative two-dimensional quantum gravity. Nuclear Physics B, 1991, 348(3): 435-456.
		
		\bibitem{dn}Do N, Norbury P. Topological recursion on the Bessel curve. Communications in Number Theory and Physics, 2018, 12(1):53-73.
		
		\bibitem{dyz} Dubrovin B, Yang D, Zagier D. On tau-functions for the KdV hierarchy. Selecta Mathematica. 2021, 27:1-47.
		
		\bibitem{eop}Eskin A, Okounkov A, Pandharipande R. The theta characteristic of a branched covering. Adv. Math., 2008, 217: 873-888.
		
		\bibitem{eo}Eynard B, Orantin N. Invariants of algebraic curves and topological expansion. Communications in Number Theory \& Physics, 2007, 1(2):347-452.
		
		\bibitem{fkn}Fukuma M, Kawai H, Nakayama R. Continuum Schwinger-Dyson Equations and Universal Structures in Two-Dimensional Quantum Gravity. International Journal of Modern Physics A, 1991, 6(08):1385-1406.
		
		\bibitem{gkl}Giacchetto A, Kramer R, Lewa\'nski D. A new spin on Hurwitz theory and ELSV via theta characteristics. arXiv preprint arXiv:2104.05697, 2021.
		
		\bibitem{gw}Gross D J, Witten E. Possible third-order phase transition in the large-$N$ lattice gauge theory. Physical Review D Particles \& Fields, 1980, 21(2).
		
		\bibitem{gs}Gukov S, Su{\l}kowski P. A-polynomial, B-model, and quantization. Journal of High Energy Physics, 2012, 2012(2):70.
		
		\bibitem{hb}Harnad J, Balogh F. Tau Functions and their Applications. Cambridge Monographs on Mathematical Physics. Cambridge University Press, 2021.
		
		\bibitem{hh}Hoffman P N, Humphreys J F. Projective representations of the symmetric groups: Q-functions and shifted tableaux. Oxford Mathematical Monographs. Clarendon Press, 1992.
		
		\bibitem{jwy}Ji C, Wang Z, Yang C. Kac-Schwarz Operators of Type $B$, Quantum Spectral Curves, and Spin Hurwitz Numbers. Journal of Geometry and Physics, 2023, 189: 104831.
		
		\bibitem{jm}Jimbo M, Miwa T. Solitons and Infinite-Dimensional Lie Algebras. Publications of the Research Institute for Mathematical Sciences, 1983, 1983(19): 943-1001.
		
		\bibitem{ks}Kac V, Schwarz A. Geometric interpretation of the partition function of 2D gravity. Physics Letters B, 1991, 257(3-4):329-334.
		
		\bibitem{kv}Kac V, van de Leur J. Polynomial tau-functions of BKP and DKP hierarchies. Journal of Mathematical Physics, 2019, 60(7):071702.
		
		\bibitem{kn}Knudsen F F. The projectivity of the moduli space of stable curves, II: The stacks $M_{g,n}$. Mathematica Scandinavica, 1983, 52(2): 161-199.
		
		\bibitem{kon}Kontsevich M. Intersection theory on the moduli space of curves and the matrix Airy function. Communications in Mathematical Physics, 1992, 147(1): 1-23.
		
		\bibitem{la}Laughlin R B. A Different Universe: Reinventing Physics from the Bottom Down. Basic Books, New York, 2005.
		
		\bibitem{ly}Liu X, Yang C. $Q$-Polynomial expansion for Br\'ezin-Gross-Witten tau-function. Advances in Mathematics, 2022, 404: 108456.
		
		\bibitem{mac}MacDonald I G. Symmetric functions and Hall polynomials. 2nd edition. Clarendon Press, 1995.
		
		\bibitem{mm}Mironov A, Morozov A. Superintegrability of Kontsevich matrix model. European Physical Journal C, 2021, 81(3).
		
		\bibitem{mms}Mironov A, Morozov A, Semenoff G. Unitary matrix integrals in the framework of Generalized Kontsevich Model. I. Br\'ezin-Gross-Witten Model. International Journal of Modern Physics A, 1996, 11(28):5031-5080.
		
		\bibitem{no}Norbury P. A new cohomology class on the moduli space of curves. Geom. Topol., 2023, 27(7): 2695-2761.
		
		\bibitem{or}Orlov A Y. Hypergeometric functions related to Schur $Q$-polynomials and BKP equation. Theoretical \& Mathematical Physics, 2003, 137(2):1574-1589.
		
		\bibitem{sa}Sato M. Soliton Equations as Dynamical Systems on an Infinite Dimensional Grassmann Manifold. RIMS Kokyuroku, 1981, 439:30-46.
		
		\bibitem{sch}Schur J. \"Uber die Darstellung der symmetrischen und der alternierenden Gruppe durch gebrochene lineare Substitutionen. Journal F\"ur Die Reine Und Angewandte Mathematik, 1911, 1911(139):155-250.
		
		\bibitem{sc}Schwarz A. On the solutions to the string equation. Modern Physics Letters A, 2011, 06(29): 2713-2725.
		
		\bibitem{sw}Segal G, Wilson G. Loop Groups and Equations of KdV Type. Publications Math\'ematiques de l'IH\'ES, 1985, 61(1):5-65.
		
		\bibitem{va}van de Leur J. The Adler-Shiota-van Moerbeke formula for the BKP hierarchy. Journal of Mathematical Physics, 1995, 36:4940-4951.
		
		\bibitem{wy}Wang Z, Yang C. BKP Hierarchy, Affine Coordinates, and a Formula for Connected Bosonic $N$-Point Functions. Letters in Mathematical Physics, 2022, 112:62.
		
		\bibitem{wz}Wang Z, Zhou J. A formalism of abstract quantum field theory of summation of fat graphs. arXiv preprint arXiv: 2108.10498, 2021.
		
		\bibitem{wz2}Wang Z, Zhou J. Topological 1D Gravity, KP Hierarchy, and Orbifold Euler Characteristics of $\overline{\cM}_{g,n}$. arXiv:2109.03394, 2021.
		
		\bibitem{wit}Witten E. Two-dimensional gravity and intersection theory on moduli space. Surveys in differential geometry, 1990, 1(1): 243-310.
		
		\bibitem{yz}Yang D, Zhang Q. On The Hodge-BGW Correspondence. Communications in
Number Theory and Physics 18.3 (2024): 611-651.
		
		\bibitem{yo}You Y. Polynomial solutions of the BKP hierarchy and projective representations of symmetric groups. Infinite-dimensional Lie algebras and groups, Luminy-Marseille, Adv. Ser. Math. Phys., Volume 7: 449-464.
		
		\bibitem{zhou2}Zhou J. Topological recursions of Eynard-Orantin type for intersection numbers on moduli spaces of curves. Letters in Mathematical Physics, 2013, 103(11): 1191-1206.
		
		\bibitem{zhou9}Zhou J. Quantum deformation theory of the Airy curve and mirror symmetry of a point. arXiv preprint arXiv:1405.5296, 2014.
		
		\bibitem{zhou1}Zhou J. Emergent geometry and mirror symmetry of a point. arXiv preprint arXiv:1507.01679, 2015.
		
		\bibitem{zhou3}Zhou J. K-Theory of Hilbert schemes as a formal quantum field theory. arXiv preprint arXiv:1803.06080, 2018.
		
		\bibitem{zhou10}Zhou J. Fat and thin emergent geometries of Hermitian one-matrix models. arXiv preprint arXiv:1810.03883, 2018.
		
		\bibitem{zhou8}Zhou J. Emergent Geometry of Matrix Models with Even Couplings. arXiv preprint arXiv:1903.10767, 2019.
		
		\bibitem{zhou5}Zhou J. Grothendieck's Dessins d'Enfants in a Web of Dualities. arXiv preprint arXiv: 1905.10773, 2019.
		
		\bibitem{zhou7}Zhou J. Grothendieck's Dessins d'Enfants in a Web of Dualities. II. arXiv preprint arXiv: 1907.00357, 2019.
		
		
		
	\end{thebibliography}
\end{document}